\documentclass[lettersize,journal]{IEEEtran}
\usepackage{pdfpages}
\usepackage{array}
\usepackage[caption=false,font=normalsize,labelfont=sf,textfont=sf]{subfig}
\usepackage{stfloats}
\usepackage{url}
\usepackage{verbatim}
\usepackage{cite}
\usepackage{amssymb,amsfonts,amsthm,amsmath}
\usepackage{pifont}
\usepackage{bbding}
\usepackage{bm}
\usepackage{graphicx}
\usepackage{textcomp}
\usepackage{microtype}
\usepackage{xcolor}
\usepackage{comment}
\usepackage{float}
\usepackage{mathtools}
\usepackage[ruled,linesnumbered]{algorithm2e}

\SetKw{Continue}{continue}
\SetKw{Break}{break}
\SetKwInOut{Input}{Input}
\SetKwInOut{Output}{Output}

\newtheorem{theorem}{Theorem}
\newtheorem{lemma}{Lemma}

\newtheorem{definition}{Definition}
\newtheorem{example}{Example}

\newtheorem*{problem*}{Problem}
\newtheorem{problem}{Problem}
\newtheorem{statement}{Statement}
\newtheorem*{statement*}{Statement}
\theoremstyle{definition}
\newtheorem*{superior*}{Superior Case}
\newtheorem*{inferior*}{Inferior Case}

\DeclareMathOperator*{\argmax}{arg\,max}
\DeclareMathOperator*{\argmin}{arg\,min}

\newcommand{\diff}{\color{black}}
\newcommand{\mypath}{\lambda}       %path
\newcommand{\pv}{\pi}               %path vertex
\newcommand{\eseq}{\varepsilon}     %execution sequence
\newcommand{\chain}{\lambda}    %chain or generalized path, to consistent with RTSS2022
\graphicspath{{./graphics/}}

\begin{document}

\title{Multi-Path Bound for DAG Tasks}

\author{Qingqiang~He,
        Nan~Guan*,
        Shuai~Zhao,
        and Mingsong~Lv
        % <-this % stops a space
\thanks{Qingqiang He is with the School of Computing and Information Technology,
Great Bay University, China. E-mail: heqq@gbu.edu.cn.}% <-this % stops a space
\thanks{Nan Guan is with the Department of Computer Science, City University of Hong Kong, China. E-mail: nanguan@cityu.edu.hk.}% <-this % stops a space
\thanks{Shuai Zhao is with the School of Computer Science and Engineering, Sun Yat-sen University, China. E-mail: zhaosh56@mail.sysu.edu.cn.}% <-this % stops a space
\thanks{Mingsong Lv is with the Department of Computing, The Hong Kong Polytechnic University, China. E-mail: mingsong.lyu@polyu.edu.hk.}% <-this % stops a space
\thanks{Manuscript received XX XX, 202X; revised XX XX, 202X. \emph{(*Corresponding author: Nan Guan.)}}
}

\markboth{Journal of \LaTeX\ Class Files,~Vol.~XX, No.~XX, XX~202X}%
{He \MakeLowercase{\textit{et al.}}: Multi-Path Bound for DAG Tasks}


\maketitle

\begin{abstract}
This paper studies the response time bound of a DAG (directed acyclic graph) task.
Recently, the idea of using multiple paths to bound the response time of a DAG task, instead of using a single longest path in previous results, was proposed and led to the so-called \emph{multi-path bound}.
Multi-path bounds can greatly reduce the response time bound and significantly improve the schedulability of DAG tasks.
This paper derives a new multi-path bound and proposes an optimal algorithm to compute this bound.
We further present a systematic analysis on the dominance and the sustainability of three existing multi-path bounds and the proposed multi-path bound.
Our bound theoretically dominates and empirically outperforms all existing multi-path bounds.
What's more, the proposed bound is the only multi-path bound that is proved to be self-sustainable.
\end{abstract}


\section{Introduction}
\label{sec:intro}
\IEEEPARstart{A}{s} multi-cores are becoming the mainstream of real-time systems for performance and energy efficiency, real-time applications, such as those in automotive, avionics and industrial domains, tend to be more complex to realize their functionalities.
The DAG (directed acyclic graph) task model is widely used to represent the complex structures of parallel real-time tasks.
For example, in the autonomous driving system, the processing chain from perception to control can be modeled as a sporadic DAG task \cite{verucchi2020latency, liu2021real}.
A large body of research works on real-time scheduling and analysis of DAG tasks have been proposed in recent years \cite{li2014analysis, casini2018partitioned, nasri2019response, voudouris2021federated, osborne2022minimizing, dai2022response, baruah2022ilp}, where a fundamental problem is how to upper-bound the response time of a DAG task executing on a parallel computing platform.
{\diff Response time bound is crucial to guarantee the correctness of time-critical embedded systems, where deadline misses may cause catastrophic consequences \cite{laplante2004real}.
}

Traditionally, researchers utilize the total workload and a single longest path to upper-bound the response time of a DAG task, such as the response time bounds in \cite{graham1969bounds, melani2015response, sun2021calculating, sun2017real, jiang2017semi, han2019response}.
These bounds generally assume that vertices not in the longest path do not execute in parallel with the longest path.
However, in real executions, many vertices not in the longest path can actually execute in parallel with the execution of the longest path.
Therefore, these bounds that rely on a single longest path are rather pessimistic in most cases.
Recently, works that utilize the total workload and multiple long paths to upper-bound the response time of a DAG task were proposed \cite{he2022bounding, he2023degree, ueter2023parallel}.
We call these bounds that use multiple paths of the DAG task as \emph{multi-path bounds}.
In contrast to bounds that rely on a single longest path, multi-path bounds, which utilize the information of multiple parallel paths to analyze the execution behavior of DAG tasks, can inherently leverage the parallel power of multi-cores.
Multi-path bounds can greatly reduce the response time bound of a DAG task and significantly improve the schedulability of DAG tasks \cite{he2022bounding, he2023degree, ueter2023parallel}.

This paper derives a new multi-path bound.
Computing a multi-path bound needs a list of paths (called generalized path list) in the DAG task.
The existing multi-path bound in \cite{he2022bounding} has the constraint that its generalized path list must include the longest path of the DAG task.
The underlying insight of multi-path bounds is that the parallel computing of multi-cores leads to the parallel execution of multiple paths in the DAG task; this phenomenon, in turn, is exploited to reduce the response time bound in the multi-path based analysis method.
It feels natural that these parallel-executing multiple paths should have nothing to do with whether they include the longest path or not.
In this paper, we lift this constraint by generalizing the concepts and lemmas of \cite{he2022bounding}, thus allowing an arbitrary generalized path list for the computation of the proposed bound, making it the most elegant and exquisite multi-path bound ever since.
Lifting this constraint also allows us to search for an optimal generalized path list such that the multi-path bound can be minimized ({\diff the example in Fig. \ref{fig:ill_nested} and Section \ref{sec:problem} illustrates the intuition of why lifting the constraint of the longest path can reduce the response time bound}).
This paper proposes an optimal algorithm for computing the generalized path list through a novel reduction to the minimum-cost flow problem~\cite{ahuja1995network}.

This paper further presents a thorough analysis on the dominance among three existing multi-path bounds \cite{he2022bounding, he2023degree, ueter2023parallel} and the proposed multi-path bound, which can serve as the guidance for practitioners choosing these multi-path bounds.
We show that the proposed bound dominates all three existing multi-path bounds and Graham's bound \cite{graham1969bounds}; the three existing multi-path bounds do not dominate each other.
This paper also investigates the self-sustainability of multi-path bounds.
Concerning the problem considered here, self-sustainability intuitively means that the response time bound of a DAG task with smaller WCETs (worst-case execution time) is no larger than the response time bound of the same DAG task but with larger WCETs.
Self-sustainability is an important aspect for schedulability tests and response time bounds in real-time scheduling \cite{baker2009sustainable}, and is particularly critical for the design and dynamic updates of real-time systems \cite{yi2019design}.
We show that the proposed bound is the only multi-path bound that is proved to be self-sustainable; the bounds in \cite{he2022bounding, ueter2023parallel} are not self-sustainable and the self-sustainability of the bound in \cite{he2023degree} is still open.

Experiments demonstrate that the proposed bound has the best performance among all existing multi-path bounds and the schedulability of DAG tasks based on the proposed multi-path bound outperforms the state-of-the-art DAG scheduling approaches by a large margin.
In summary, this paper presents four contributions.
\begin{itemize}
  \item A new multi-path bound is proposed (Section \ref{sec:bound}).
  \item An optimal algorithm is provided for computing the proposed bound (Section \ref{sec:computation}).
  \item The dominance among multi-path bounds is analyzed: our bound dominates all three existing multi-path bounds and Graham's bound (Section \ref{sec:dominance}).
  \item The sustainability of multi-path bounds is investigated: our bound is the only multi-path bound that is proved to be self-sustainable (Section \ref{sec:sustain}).
\end{itemize}

\section{Related Work}
\label{sec:related}
In \cite{graham1969bounds}, Graham developed a classic response time bound using the total workload and the length of the longest path in a DAG task.
Graham's bound assumes a DAG task executing on an identical multi-core platform under a work-conserving scheduler, which is also the assumption of this paper.

Over the years, researchers sought to improve Graham's bound in mainly three directions.
The first direction is to change the task model of Graham's bound (i.e. the DAG task model).
Some works extended Graham's bound to other task models, such as the conditional DAG task model \cite{melani2015response}, graph model with loop structures \cite{sun2021calculating}, graph models of OpenMP workload \cite{serrano2015timing, sun2017real}, and DAG models with mutually exclusive executions \cite{liang2023response}.
The second direction is to change the computing model of Graham's bound (i.e. the identical multi-core platform).
Graham's bound has been adapted to uniform \cite{jiang2017semi}, heterogeneous \cite{han2019response, lin2022type} and unrelated \cite{voudouris2021bounding} multi-core platforms.
The third direction is to change the scheduling model of Graham's bound (i.e. the work-conserving scheduler).
\cite{he2019intra, zhao2022dag, he2024longer} improved Graham's bound by enforcing priority orders or precedence constraints among vertices, so their results are not general to all work-conserving scheduling algorithms.
\cite{voudouris2017timing, chen2019timing} developed scheduling algorithms based on statically assigned vertex execution orders, which are no longer work-conserving.
In summary, none of these works improves Graham's bound under the same setting (i.e. the DAG task model, the identical multi-core platform and the work-conserving scheduler) as the original work \cite{graham1969bounds}.
In other words, all of these bounds degrade to Graham's bound under the setting of \cite{graham1969bounds}.

Under the same setting as \cite{graham1969bounds}, \cite{he2022bounding} presented the first improvement over Graham's bound using the total workload and the lengths of multiple paths in the DAG task, instead of the longest path in Graham's bound.
The multi-path bound in \cite{he2022bounding} significantly reduces the response time bound and improves the schedulability of DAG tasks. After \cite{he2022bounding}, following the idea of using multiple paths to bound the response time, multi-path bounds in \cite{he2023degree, ueter2023parallel} were proposed.
\cite{he2023degree} utilized the degree of parallelism to derive its bound and to identify the multiple paths for computing its bound.
\cite{ueter2023parallel} utilized vertex-level priorities to help derive its bound. Vertex-level priority simplified the derivation of the multi-path bound but complicated the scheduling of the DAG task.

\section{System Model}
\label{sec:model}
\subsection{Task Model}
\label{sec:dag}

\begin{figure}[t]
\centering
\subfloat[]{
    \includegraphics[width=0.35\linewidth]{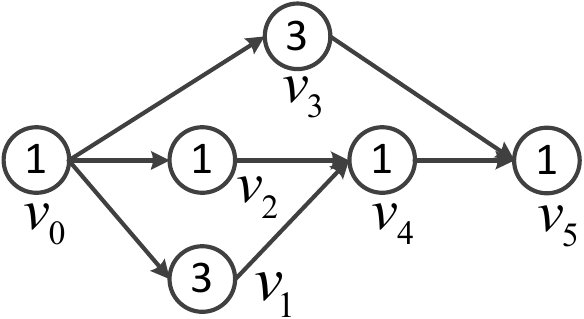}
    \label{fig:dag}
}
\hfil
\subfloat[]{
    \includegraphics[width=0.41\linewidth]{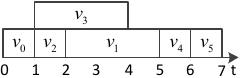}
    \label{fig:sequence}
}
\caption{(a) A DAG task example. (b) An execution sequence of Fig. \ref{fig:dag}. }
\label{fig:example}
\end{figure}

A parallel real-time task is modeled as a directed acyclic graph $G = (V, E)$, where $V$ is the set of vertices and $E\subseteq  V \times V $ is the set of edges.
Each vertex $v\in V$ represents a piece of sequentially executed workload and has a WCET (worst-case execution time) $c(v)$.
An edge $(v_i, v_j)\in E$ represents the precedence constraint between $v_i$ and $v_j$, which means that $v_j$ can only start its execution after $v_i$ completes its execution.
A vertex with no incoming edges is called a \emph{source vertex} and a vertex with no outgoing edges is called a \emph{sink vertex}.
Without loss of generality, we assume that $G$ has exactly one source  (denoted as $v_{src}$) and one sink (denoted as $v_{snk}$).
If $G$ has multiple source or sink vertices, we add a dummy source or sink vertex with zero WCET to comply with the assumption.

A \emph{path} $\mypath$ is a set of vertices $(\pv_0, \cdots, \pv_k)$ such that $\forall i\in [0,k-1]$: $(\pv_i, \pv_{i+1})\in E$.
The length of a path is the total workload in this path and is defined as $len(\mypath) \coloneqq \sum_{\pv_i\in \mypath} c(\pv_i)$.
A \emph{complete path} is a path starting from the source vertex and ending at the sink vertex.
The \emph{longest path} is a complete path with the largest length among all paths of $G$.
For a vertex set $U \subseteq V$, we define $vol(U) \coloneqq \sum_{v\in U} c(v)$.
The \emph{length} of the longest path in $G$ is denoted as $len(G)$.
The \emph{volume} of $G$ is the total workload in $G$ and is defined as $vol(G) \coloneqq \sum_{v\in V} c(v)$.

If there is an edge $(u,v)$, we say that $u$ is a \emph{predecessor} of $v$, and $v$ is a \emph{successor} of $u$.
If there is a path starting from $u$ and ending at $v$, we say that $u$ is an \emph{ancestor} of $v$ and $v$ is a \emph{descendant} of $u$.
The sets of predecessors, successors, ancestors and descendants of $v$ are denoted as $pre(v)$, $suc(v)$, $anc(v)$ and $des(v)$, respectively.
A \emph{generalized path} $\chain=(\pv_0, \cdots, \pv_k)$ is a set of vertices such that $\forall i\in [0,k-1]$: $\pv_i$ is an ancestor of $\pv_{i+1}$.
By definition, a path is a generalized path, but a generalized path is not necessarily a path.
Similar to paths, the length of a generalized path $\chain$ is defined as $len(\chain) \coloneqq \sum_{\pv_i\in \chain} c(\pv_i)$.

\subsection{Scheduling Model}
\label{sec:schedule}
The task $G$ is scheduled to execute on a computing platform with $m$ identical cores.
A vertex $v$ is said to be \emph{eligible} if all its predecessors have finished execution and thus $v$ can be immediately executed when there are available cores.
For a scheduling algorithm, the \emph{work-conserving} property means that an eligible vertex must be executed if there are available cores. We do not restrict the scheduling algorithm, as long as it satisfies the work-conserving property.

At runtime, vertices of $G$ execute at certain time points on certain cores under the decision of the scheduling algorithm.
An \emph{execution sequence} of $G$ describes which vertex executes on which core at every time point.
In an execution sequence, the \emph{start time} $s(v)$ and \emph{finish time} $f(v)$ are the time point when $v$ first starts its execution and completes its execution, respectively.
Without loss of generality, we assume the source vertex of $G$ starts its execution at time $0$.
The \emph{response time} $R$ of $G$ in an execution sequence is defined to be the time when the sink vertex finishes its execution, i.e. $f(v_{snk})$.

\begin{example}\label{exp:dag}
Fig. \ref{fig:dag} shows a DAG task $G$ and Fig. \ref{fig:sequence} shows a possible execution sequence of $G$ under a work-conserving scheduler.
The number inside vertices is the WCET of vertices.
%The source vertex and the sink vertex are $v_0$ and $v_5$, respectively.
The longest path is $\mypath=(v_0, v_1, v_4, v_5)$, and $len(G) = len(\mypath)=6$.
%For vertex set $U=\{v_1, v_2\}$, $vol(U)=4$. 
The volume of the DAG is $vol(G)=10$.
%For vertex $v_4$, $pre(v_4)=\{v_1, v_2\}$, $suc(v_4)=\{v_5\}$, $anc(v_4)=\{v_0, v_1, v_2\}$, $des(v_4)=\{v_5\}$.
$\chain=(v_0, v_2, v_5)$ is a generalized path. Note that by definition, $\chain$ is not a path, because $(v_2, v_5)$ is not an edge in the DAG.
%In Fig. \ref{fig:sequence}, the start time and finish time of $v_1$ are $s(v_1)=2$ and $f(v_1)=5$, respectively.
%The response time of $G$ in this execution sequence is 7.
\end{example}

   %from edge

\section{New Multi-Path Bound}
\label{sec:bound}
\begin{table}[t]
\centering
\caption{{\diff Notations Used in Section \ref{sec:bound} }}
\label{tab:notation}
\begin{tabular}{ll}
\hline
\textbf{Notation}   & \textbf{Description} \\
\hline
$\varepsilon$   & an execution sequence \\
$s(v)$          & the start time of vertex $v$ \\
$f(v)$          & the finish time of vertex $v$ \\
$(P_i)_0^{k}$   & the set of cores $(P_0, P_1, \cdots, P_{k})$ \\

\hline
$\mypath^*$             & the critical path \\
$\mypath^+$             & the restricted critical path \\
$\omega$, $\eta$        & a virtual path \\
$len(\omega)$           & the length of virtual path $\omega$ \\
$(\omega_i)_0^k$        & a virtual path list \\
$(\mypath_i)_0^k$       & a generalized path list \\

\hline
$U_\varepsilon$        & the projection of vertex set $U$ regarding $\varepsilon$ \\
$\mypath_\varepsilon$  & the projection of path $\mypath$ regarding $\varepsilon$  \\
$\Delta(U)$            & the workload reduction of vertex set $U$ in $\varepsilon$ \\
\hline
\end{tabular}
\end{table}

This section presents a new response time bound for a DAG task using multiple long paths. The new multi-path bound eliminates the constraint of the longest path and does not depend on vertex-level priorities, making it the most elegant and exquisite multi-path bound ever since.

The derivation of the proposed bound is achieved by generalizing the concepts and lemmas in \cite{he2022bounding}.
Due to the page limits, we cannot include the full proof in this paper.
Since the proof for our bound shares most of the abstractions and concepts with the proof of \cite{he2022bounding}, we only highlight the differences between the two proofs in this section\footnote{For the full proof, please refer to the supplementary material attached in the end of this paper.}.
{\diff Table~\ref{tab:notation} summarizes the notations used in this section; the precise definitions of these notations are in the supplementary material.}
Although we tried our best to be as clear as possible, Section~\ref{sec:bound} is not self-contained.

Our method requires multiple long paths to compute the proposed bound. These multiple long paths are formally characterized by a \emph{generalized path list} in Definition \ref{def:path_list}.
\begin{definition}[Generalized Path List]\label{def:path_list}
A generalized path list is a set of disjoint generalized paths $(\chain_i)_0^k$ ($k \ge 0$), i.e.,
$$\forall i, j \in [0, k],\ \chain_i \cap \chain_j = \varnothing$$
\end{definition}
In the above definition , $(\chain_i)_0^k$ is the compact representation of $(\chain_0, \cdots, \chain_k)$.
\begin{example}\label{exp:path_list}
For the DAG task in Fig. \ref{fig:dag}, $(\chain_i)_0^1$ with $\chain_0=(v_0, v_1, v_4, v_5)$ and $\chain_1=(v_3)$ is a generalized path list, where the first generalized path $\chain_0$ is the longest path.
As an another example, $(\chain_i)_0^2$ with $\chain_0=(v_0, v_1)$, $\chain_1=(v_2, v_4)$ and $\chain_2=(v_3, v_5)$ is also a generalized path list.
\end{example}

\textbf{The first difference} between the two proofs lies in the generalization of Lemma 4 of \cite{he2022bounding}.
For an execution sequence, both our proof and the proof in \cite{he2022bounding} conduct analysis on the workload situated in the time interval during which the generalized path list is executing.
Since the generalized path list used in \cite{he2022bounding} requires the first generalized path to be the longest path of the DAG task (note that the longest path starts from the source vertex and ends at the sink vertex), \cite{he2022bounding} analyzes the workload in time interval $[s(v_{src}), f(v_{snk})]$, i.e. $[0, f(v_{snk})]$.
However, our method does not have this requirement. Next, we introduce notations to define the time interval during which the generalized path list is executing.

For an execution sequence $\varepsilon$ and a generalized path list $(\chain_i)_0^k$, $k \in [0, m-1]$, we define
\begin{equation}\label{equ:first_list}
v_\emph{fst} = \argmin_{u\in (\chain_i)_0^k} \{s(u)\}
\end{equation}
\begin{equation}\label{equ:last_list}
v_\emph{lst} = \argmax_{u\in (\chain_i)_0^k} \{f(u)\}
\end{equation}
Intuitively, for vertices of this generalized path list, $v_\emph{fst}$ is the first vertex to start its execution and $v_\emph{lst}$ is the last vertex to finish its execution in $\varepsilon$.

\begin{lemma}[Corresponding to Lemma 4 of \cite{he2022bounding}]
\label{lem:restricted}
$(\mypath_i)_0^k$, $k \in [0, m-1]$, is a generalized path list.
$\varepsilon$ is a regular execution sequence regarding $(\mypath_i)_0^k$.
$\mypath^+$ is the restricted critical path of $(\mypath_i)_0^k$ in $\varepsilon$.
$\mypath^+_\varepsilon$ is the projection of $\mypath^+$ in $\varepsilon$.
There exists a virtual path $\eta$ in $\varepsilon$ satisfying all the following three conditions.
\begin{enumerate}
  \item [(i)] $\forall v \in \eta$, $v \notin (\mypath_i)_0^k$;
  \item [(ii)] $\forall v \in \eta$, $v$ executes on $(P_i)_0^{k}$;
  \item [(iii)] $len(\mypath^+_\varepsilon)+ len(\eta)=f(v_\text{lst})-s(v_\text{fst})$.
\end{enumerate}
\end{lemma}
It is not required that the first generalized path is the longest path in our bound. So we focus on the workload in time interval $[s(v_\emph{lst}), f(v_\emph{fst})]$.
In Lemma \ref{lem:restricted}, Condition (iii) is modified compared to \cite{he2022bounding}, but the underlying reasoning is unchanged.

\textbf{The second difference} is in Lemma 8 and Lemma 9 of \cite{he2022bounding}.
For the execution sequence under analysis, the execution of the first generalized path is important for the proof.
In the following, regarding a generalized path list $(\mypath_i)_0^k$, we also use $\mypath$ to denote $\mypath_0$ for conciseness.
The definitions of time interval $H$ and $Z$ are modified as follows.

For an execution sequence $\varepsilon$ and a generalized path list $(\mypath_i)_0^k$, $k \in [0, m-1]$, we define
\begin{itemize}
  \item $H$: time interval in $[s(v_\emph{fst}), f(v_\emph{lst})]$ during which $\exists \pv_i \in \mypath$, $\pv_i$ is executing;
  \item $Z$: time interval in $[s(v_\emph{fst}), f(v_\emph{lst})]$ during which $\forall \pv_i \in \mypath$, $\pv_i$ is \emph{not} executing.
\end{itemize}

In \cite{he2022bounding}, $H$ is defined to be the time interval where the longest path is executing.
Since in our method, the first generalized path does not have to be the longest path, $H$ is modified to be the time interval where the first generalized path in $(\mypath_i)_0^k$ is executing.
By the definition of $H$ and $Z$, we have
\begin{equation}\label{equ:hz}
|H|+|Z|=f(v_\emph{lst})-s(v_\emph{fst})
\end{equation}

In an execution sequence $\varepsilon$, for a vertex set $U$, let $\Delta(U)$ denote the amount of workload reduction of $U$ in $\varepsilon$. For example, in Fig. \ref{fig:dag}, let $U=\{v_1, v_2\}$. The total workload of $U$ is $vol(U)=4$. Suppose in an execution sequence $\varepsilon$, the execution times of $v_1$ and $v_2$ are 2 and 1, respectively. Then $\Delta(U)=4-(2+1)=1$.

\begin{lemma}[Corresponding to Lemma 8 of \cite{he2022bounding}]
\label{lem:workload_volume}
$vol(W') \ge \sum_{i=1}^{k} len(\mypath_i) -\Delta(V)- (len(G)-len(\mypath))$.
\end{lemma}
Compared to \cite{he2022bounding}, the item $len(G)-len(\mypath)$ is added in Lemma \ref{lem:workload_volume}. This is also due to the fact that the first generalized path does not have to be the longest path in our method.
When the first generalized path $\mypath_0=\mypath$ is the longest path, $len(G)-len(\mypath)=0$, and Lemma \ref{lem:workload_volume} degrades to Lemma 8 of \cite{he2022bounding}.
For the proof of Lemma \ref{lem:workload_volume}, the major difference is (\ref{equ:workload_volume}).
\begin{equation}\label{equ:workload_volume}
len(\eta)-|Z| = \ge -\Delta(\mypath)-(len(G)-len(\mypath))
\end{equation}
Equation (\ref{equ:workload_volume}) is mainly a result of (\ref{equ:hz}) and Condition (iii) in Lemma \ref{lem:restricted}.

\begin{lemma}[Corresponding to Lemma 9 of \cite{he2022bounding}]
\label{lem:workload}
$\varepsilon$ is an execution sequence.
$(\mypath_i)_0^k$, $k \in [0, m-1]$, is a generalized path list. Also $\mypath \coloneqq \mypath_0$.
For any complete path $\mypath'$ of $G$, there is a virtual path list $(\omega_i)_0^k$ where $\omega_0=\mypath'_{\varepsilon}$, satisfying the following condition.
\begin{equation}\label{equ:workload}
\sum_{i=1}^{k} len(\omega_i) \ge \sum_{i=1}^{k} len(\mypath_i)-\Delta(V)-(len(G)-len(\mypath))
\end{equation}
\end{lemma}
Same as Lemma \ref{lem:workload_volume}, compared to \cite{he2022bounding}, the item $len(G)-len(\mypath)$ is also added in (\ref{equ:workload}).
When the first generalized path in $(\mypath_i)_0^k$ is the longest path of $G$, we have $len(\mypath)=len(G)$. Equation (\ref{equ:workload}) degrades to
\begin{equation}\label{equ:workload1}
\sum_{i=1}^{k} len(\omega_i) \ge \sum_{i=1}^{k} len(\mypath_i)-\Delta(V)
\end{equation}
which is the equation in Lemma 9 of \cite{he2022bounding}.
With respect to the corresponding lemmas in \cite{he2022bounding}, the modification in Lemma \ref{lem:workload} is a direct result of the modification in Lemma \ref{lem:workload_volume}.

\textbf{The Third difference} is in Lemma 10 of \cite{he2022bounding}.
With the modifications in Lemmas 1-3, we are ready to state Lemma \ref{lem:dag_bound}, which is an important lemma for deriving the proposed bound.

\begin{lemma}[Corresponding to Lemma 10 of \cite{he2022bounding}]
\label{lem:dag_bound}
Given a generalized path list $(\mypath_i)_0^k$ ($k \in [0, m-1]$),
the response time $R$ of DAG $G$ scheduled by work-conserving scheduling on $m$ cores is bounded by:
\begin{equation}\label{equ:dag_bound}
R\le len(G)+\frac{vol(G)-\sum_{i=0}^{k} len(\mypath_i)}{m-k}
\end{equation}
\end{lemma}
The proof of Lemma \ref{lem:dag_bound} can be divided into two parts: the execution-level part and the graph-level part.
The execution-level part deals with the abstractions specific to an execution sequence, such as the execution time of vertices (recall that in an execution sequence, the execution time of a vertex may be less than its WCET).
The reasoning of this part is largely the same as that of \cite{he2022bounding}, which uses Lemma \ref{lem:workload} to derive that the response time $R$ of an execution sequence $\eseq$ is bounded by
\begin{equation}\label{equ:dag_bound1}
R \le len(\mypath^*)+\frac{vol(G)-len(\mypath^*)-\sum_{i=0}^{k} len(\mypath_i)+len(G)}{m-k}
\end{equation}
where $\mypath^*$ is the critical path\footnote{A critical path \cite{he2019intra} is a complete path specific to an execution sequence.} of this execution sequence $\eseq$.

The graph-level part of this proof deals with the abstractions specific to a DAG task. This part of reasoning is new and different from \cite{he2022bounding}. The goal of this part is to derive that the bound in (\ref{equ:dag_bound}) is larger than or equal to the bound in (\ref{equ:dag_bound1}).
Let $B_0$ denote the bound in (\ref{equ:dag_bound}) and $B_1$ denote the bound in (\ref{equ:dag_bound1}).
\begin{align*}
&B_0 \begin{aligned}[t]
    &-B_1 = len(G)+\frac{vol(G)-\sum_{i=0}^{k} len(\mypath_i)}{m-k}- \\
    &(len(\mypath^*)+\frac{vol(G)-len(\mypath^*)-\sum_{i=0}^{k} len(\mypath_i)+len(G)}{m-k}) \\
\end{aligned} \\
&=len(G)-(len(\mypath^*)+\frac{len(G)-len(\mypath^*)}{m-k}) \\
&=len(G)-len(\mypath^*)-\frac{len(G)-len(\mypath^*)}{m-k} \\
&=(len(G)-len(\mypath^*))(1-\frac{1}{m-k})
\end{align*}
Since $len(G) \ge len(\mypath^*)$, we have $B_0-B_1 \ge 0$, which means $B_0 \ge B_1$.

In summary of these three differences in the derivation, the proposed bound is presented in Theorem \ref{thm:our_bound}.
\begin{theorem}\label{thm:our_bound}
Given a generalized path list $(\mypath_i)_0^k$ ($k \in [0, m-1]$), the response time $R$ of DAG $G$ scheduled by work-conserving scheduling on $m$ cores is bounded by:
\begin{equation}\label{equ:our_bound}
R\le \min \limits_{j \in [0, k]} \left\{ len(G)+\frac{vol(G)-\sum_{i=0}^{j} len(\mypath_i)}{m-j} \right\}
\end{equation}
\end{theorem}
With Lemma \ref{lem:dag_bound}, the proof of Theorem \ref{thm:our_bound} is exactly the same as the proof of Theorem 2 in \cite{he2022bounding}.
The only difference between our bound in Theorem \ref{thm:our_bound} and the bound in \cite{he2022bounding} (see Theorem~\ref{thm:rtss_bound} in Section \ref{sec:dominance}) is that we do not have the constraint that the first generalized path $\chain_0$ of $(\mypath_i)_0^k$ should be the longest path of DAG $G$. We call this as the \emph{constraint of the longest path}.

The generalized path list $(\mypath_i)_0^k$ is an indispensable part of computing the bound in (\ref{equ:our_bound}).
Lifting the constraint of the longest path opens up three opportunities.
\begin{enumerate}
  \item \emph{Tightness}. The response time bound for a DAG task can be reduced, thus the schedulability for DAG tasks can be improved. Section \ref{sec:computation} presents an optimal algorithm to compute the proposed bound.
  \item \emph{Dominance}. The dominance among multi-path bounds can be established. Section \ref{sec:dominance} shows that the proposed bound dominates all three existing multi-path bounds \cite{he2022bounding, he2023degree, ueter2023parallel}.
  \item \emph{Sustainability}. The sustainability of multi-path bounds can be analyzed. Section \ref{sec:sustain} demonstrates that the proposed bound is the only multi-path bound that is proved to be self-sustainable.
  %\item \emph{Extension}. The multi-path bound can be potentially extended to other more complex and more realistic task models, such as the conditional DAG task model. Section~\ref{sec:conditional} briefly discusses why lifting the constraint of the longest path makes this kind of extension possible.
\end{enumerate}

\section{Computation of the Bound}
\label{sec:computation}
The computation of the proposed response time bound in Theorem \ref{thm:our_bound} requires a generalized path list to be given.
This section studies how to compute this generalized path list for a DAG task.
Section \ref{sec:problem} uses an example to introduce and define the problem.
Section \ref{sec:optimal} presents an optimal algorithm for computing the generalized path list.
Section \ref{sec:overall} provides the overall method to compute the proposed bound optimally.

\subsection{The Problem}
\label{sec:problem}
This subsection first discusses how the generalized path list without the constraint of the longest path affects the response time bound, and second extracts the definition of the problem we are trying to solve.

We use an example for illustration.
Fig. \ref{fig:ill_nested} shows a DAG task $G$. The length $len(G)=4$ and the volume $vol(G)=6$. Let the number of cores $m=2$.
$(\chain_i)_0^1$ with $\chain_0=(v_0, v_3)$ and $\chain_1=(v_1)$ is a generalized path list where the first generalized path $\chain_0$ is the longest path.
For this generalized path list, the bound in (\ref{equ:our_bound}) is computed as
\begin{align*}
R   &\le \min\{4+(6-4)/2, 4+(6-4-1)/(2-1)\} \\
    &=\min\{5, 5\}=5
\end{align*}
Since the first generalized path $\chain_0$ is the longest path, the bound in \cite{he2022bounding} is also $5$.

Next we use another generalized path list without the constraint of the longest path to compute the bound. $(\chain_i)_0^1$ with $\chain_0=(v_0, v_1)$ and $\chain_1=(v_2, v_3)$ is a generalized path list.
For this generalized path list, the bound in (\ref{equ:our_bound}) is computed as
\begin{align*}
R   &\le \min\{4+(6-4)/2, 4+(6-4-2)/(2-1)\} \\
    &=\min\{5, 4\}=4
\end{align*}
Since in this case, the first generalized path $\chain_0$ is not the longest path, bound $R \le 4$ cannot be achieved by the method in \cite{he2022bounding}.

\begin{figure}[t]
\centering
\subfloat[]{
    \includegraphics[width=0.22\linewidth]{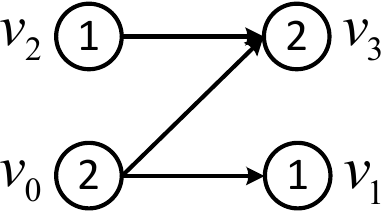}
    \label{fig:ill_nested}
}
\subfloat[]{
    \includegraphics[width=0.76\linewidth]{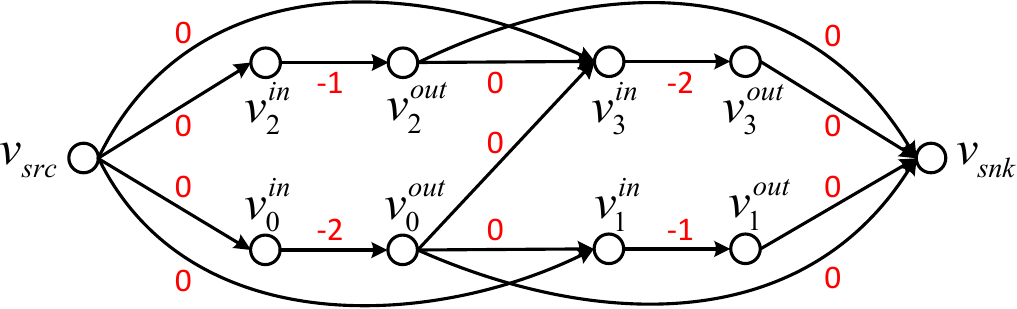}
    \label{fig:mini_cost}
}
\caption{(a) A DAG example for Section \ref{sec:problem}. (b) A flow network constructed from the DAG in Fig. \ref{fig:ill_nested}.}
\label{fig:comp_example}
\end{figure}

The example in Fig. \ref{fig:ill_nested} illustrates how lifting the constraint of the longest path may reduce the response time bound of DAG tasks.
In practice, algorithms for computing the generalized path list in \cite{he2022bounding, ueter2023parallel, he2023degree} can all be used to compute the generalized path list for the proposed response time bound in (\ref{equ:our_bound}).
We observe that all the algorithms for computing the generalized path list in \cite{he2022bounding, ueter2023parallel, he2023degree} are not optimal.
In Lemma \ref{lem:dag_bound}, for a DAG task $G$ and a specific $k \in [0, m-1]$, to minimize (\ref{equ:dag_bound}),  $\sum_{i=0}^{k} len(\mypath_i)$ should be maximized.
For a generalized path list, we call the number of generalized paths in this generalized path list as the \emph{cardinality}, and the total workload in this generalized path list as the \emph{volume}.
For $(\mypath_i)_0^k$, the cardinality is $k+1$ and the volume is $\sum_{i=0}^{k} len(\mypath_i)$.
From the context of DAG scheduling, we extract the following problem.

\begin{problem}\label{pro:path}
Given a DAG task $G=(V, E)$ and an integer $n$, how to compute a generalized path list of cardinality $n$ with the maximum volume?
\end{problem}

If $n=1$, this problem degrades to computing the longest path of a directed acyclic graph and the maximum volume is the length of the longest path $len(G)$. In this case, the problem is of time complexity $O(|V|+|E|)$.
If $n$ is equal to or larger than the width\footnote{A DAG is a partially ordered set. The width of a partially ordered set is the maximum number of mutually incomparable elements in this partially ordered set \cite{gratzer2002general}.} of DAG $G$, this problem degrades to computing the width of a directed acyclic graph and the maximum volume is the volume of the DAG $vol(G)$. In this case, the problem is of time complexity $O(|E|\sqrt{|V|})$ \cite{hopcroft1973n, garey1979computers}.

However, for $1<n<w$ (where $w$ denotes the width of DAG $G$), to the best of our knowledge, this problem is still open.
Notably, \cite{ueter2023parallel} shows that its algorithm for computing the generalized path list (i.e. Algorithm 1 of \cite{ueter2023parallel}) has an approximation ratio of $2-\frac{1}{w}$, which indicates the gap between the algorithm in \cite{ueter2023parallel} and the optimal algorithm of Problem \ref{pro:path}.

\subsection{The Optimal Algorithm}
\label{sec:optimal}
The subsection presents an optimal algorithm for Problem \ref{pro:path} by reducing it to the \emph{minimum-cost flow} problem \cite{ahuja1995network}, which has known solutions.

\vspace{1mm}\noindent
\textbf{\S~Minimum-Cost Flow Problem}

First, we introduce the minimum-cost flow problem.
A \emph{flow network} is a directed acyclic graph $G = (V, E)$, where $V$ is the set of vertices and $E\subseteq  V \times V $ is the set of edges. $G$ is with a single source vertex $v_{src} \in V$ and a single sink vertex $v_{snk} \in V$.
Each edge $(u, v) \in E$ is with a capacity $a(u, v)>0$ and a cost $c(u, v)$.

A \emph{flow} for the graph $G$ is a collection $\{f(u, v) \mid (u, v) \in E \}$ satisfying (\ref{equ:flow_value}) and (\ref{equ:flow_topology}).
\begin{equation}\label{equ:flow_value}
\forall (u, v) \in E: 0 \le f(u, v) \le a(u, v)
\end{equation}
\begin{equation}\label{equ:flow_topology}
\forall v \in V \setminus \{v_{src}, v_{snk}\}:  \sum_{u \in pre(v)} f(u, v)=\sum_{w \in suc(v)} f(v, w)
\end{equation}
where $pre(v)$ and $suc(v)$ are the set of predecessors and successors of $v$, respectively.
(\ref{equ:flow_value}) means that a flow $f(u, v)$ for an edge $(u, v)$ cannot exceed the capacity of this edge.
(\ref{equ:flow_topology}) means that for a vertex that is not the source or the sink, the total flow into this vertex equals the total flow out of this~vertex.

The \emph{amount} of a flow for $G$ is
$$\sum_{w \in suc(v_{src})} f(v_{src}, w) \text{\quad or \ } \sum_{u \in pre(v_{snk})} f(u, v_{snk})$$
In other words, the amount of flow is the total flow out of the source vertex, which is the same as the total flow into the sink vertex by (\ref{equ:flow_topology}).
The \emph{cost} of a flow for $G$ is
$$\sum_{(u, v) \in E} f(u, v) \cdot c(u, v)$$

\begin{problem}[Minimum-Cost Flow]\label{pro:flow}
Given a flow network $G=(V, E)$ and an integer $n$, how to compute a flow of amount $n$ with the minimum cost?
\end{problem}

There are many algorithms for the minimum-cost flow problem, such as the successive shortest path algorithm with time complexity $O(n(|E|+|V|log|V|))$ \cite{ahuja1995network}.

\vspace{1mm}\noindent
\textbf{\S~The Reduction}

Second, we present the reduction from Problem \ref{pro:path} to Problem \ref{pro:flow}.
From an arbitrary instance of Problem \ref{pro:path} (i.e. a DAG task $G=(V, E)$ and an integer $n$), we construct an instance of the minimum-cost flow problem by the following procedure.
\begin{enumerate}
  \item Connect ancestors. $\forall v \in V$ and $\forall u \in anc(v)$, add edge~$(u, v)$.
  \item Spilt vertices. $\forall v \in V$, replace $v$ with two vertices $v^{in}$ and $v^{out}$; add edge $(v^{in}, v^{out})$ with the cost $c(v^{in}, v^{out})=-c(v)$.
  \item Add source and sink. Add a source vertex $v_{src}$; for each $v^{in}$, add edge $(v_{src}, v^{in})$. Add a sink vertex $v_{snk}$; for each $v^{out}$, add edge $(v^{out}, v_{snk})$.
  \item Assign capacity and cost. For each edge $(u, v)$, the capacity $a(u, v)=1$. For each edge $(u, v)$ except for edges added in Step 2, the cost $c(u, v)=0$.
  \item The amount $n$ in Problem \ref{pro:flow} equals the cardinality $n$ in Problem \ref{pro:path}.
\end{enumerate}

In the constructed flow network, the number of vertices is $2|V|+2$ and the number of edges is at most $O(|V|^2)$.
Therefore, the above procedure is a polynomial reduction.

\begin{example}\label{exp:mini_cost}
Fig. \ref{fig:mini_cost} shows a flow network $G'=(V', E')$ constructed from the DAG task $G=(V, E)$ in Fig. \ref{fig:ill_nested} by the above reduction procedure. The red number beside the edges is the cost of this edge. The capacity of all edges is 1 and is not indicated in Fig. \ref{fig:mini_cost}.
\end{example}

\vspace{1mm}\noindent
\textbf{\S~The Correctness}

Third, we prove the correctness of the reduction as seen in Theorem \ref{thm:computation}.

\begin{theorem}\label{thm:computation}
Problem \ref{pro:path} can be solved in polynomial time.
\end{theorem}
\begin{proof}
We prove this by showing that the reduction from Problem \ref{pro:path} to Problem \ref{pro:flow} is correct, such that Problem \ref{pro:path} can be solved by using algorithms for Problem \ref{pro:flow}.
To prove the correctness of the above reduction, we follow the proving framework of many-one reduction by showing that the decision version of Problem \ref{pro:path} is equivalent to the corresponding decision version of Problem \ref{pro:flow}.
Specifically, we prove that Statement \ref{sta:path} is true if and only if Statement \ref{sta:flow} is true.

\begin{statement}\label{sta:path}
In DAG task $G=(V, E)$, there exists a generalized path list of cardinality $n$ with volume no less than $X$.
\end{statement}
\begin{statement}\label{sta:flow}
In flow network $G'=(V', E')$, which is constructed from DAG task $G$ by the above reduction, there exists a flow of amount $n$ with cost no greater than $-X$.
\end{statement}

\noindent
\textbf{Sufficiency:} if Statement \ref{sta:flow} is true, then Statement \ref{sta:path} is true.

Let this flow be $\{f(u, v) \mid (u, v) \in E' \}$ and the cost of this flow is no greater than $-X$.
Since the capacity of every edge in $G'$ is 1, we have that either $f(u, v)=1$ or $f(u, v)=0$ in this flow.
Now we construct the required generalized path list for $G$.
Since the amount of this flow is $n$, in the outgoing edges of $v_{src}$, there are $n$ edges with a flow of $1$.
For each of these $n$ edges, following this one unit of flow, a generalized path $\chain_i$ can be constructed:
initially, let $\chain_i=\varnothing$; whenever encountering a vertex $v^{in}_j$ with $f(v^{in}_j, v^{out}_j)=1$, put $v_j$ into $\chain_i$ until the sink vertex of $G'$ is reached.
Now, $n$ number of generalized paths are constructed.

Since each $v^{in}$ has only one outgoing edge and its capacity is $1$, by (\ref{equ:flow_topology}), there is at most one edge with a flow of $1$ among the incoming edges of $v^{in}$. Therefore, there are no common vertices between any two generalized paths, which means that these generalized paths can form a generalized path list $(\chain_i)_0^{n-1}$.
Since the cost of each edge $(v^{in}, v^{out})$ in $G'$ is the negative of $c(v)$ in $G$ and the cost of this flow is no greater than $-X$, the volume of $(\chain_i)_0^{n-1}$ is no less than $X$.

\noindent
\textbf{Necessity:} if Statement \ref{sta:path} is true, then Statement \ref{sta:flow} is true.

Let this generalized path list be $(\chain_i)_0^{n-1}$ and the volume of this generalized path list is no less than $X$.
For each $\chain_i$ in this generalized path list, let $\chain_i=(\pv_0, \cdots, \pv_j, \cdots, \pv_k)$.
Now we construct the required flow for $G'$:
first, let $f(v_{src}, \pv^{in}_0)=1$ and $f(\pv^{in}_0,\pv^{out}_0)=1$;
second, for each $\pv_j$ and $0<j \le k$, let $f(\pv^{out}_{j-1}, \pv^{in}_j)=1$ and $f(\pv^{in}_j,\pv^{out}_j)=1$;
third, let $f(\pv^{out}_k, v_{snk})=1$;
finally, let the flow for all other edges in $G'$ be 0.
Now, for each edge in $G'$, a flow for this edge is constructed.

Since $(\chain_i)_0^{n-1}$ is a generalized path list, there are no common vertices between any two generalized paths.
Therefore, for each vertex $v$ in this generalized path list, there is only one edge with a flow of $1$ in the incoming edges of $v^{in}$ and there is only one edge with a flow of $1$ in the outgoing edges of $v^{out}$.
Therefore, for each vertex $v$ in this generalized path list, we have
$\sum_{u \in pre(v^{in})} f(u, v^{in})=f(v^{in}, v^{out})=\sum_{w \in suc(v^{out})} f(v^{out}, w)=1$.
For each vertex $v$ not in this generalized path list, the total flow into $v^{in}$, the total flow out of $v^{in}$, the total flow into $v^{out}$ and the total flow out of $v^{out}$ are all $0$.
In summary, (\ref{equ:flow_topology}) holds. And obviously, (\ref{equ:flow_value}) holds.
Therefore, the above constructed flows for all edges in $G'$ can form a flow for $G'$.

Since there are $n$ number of generalized paths in $(\chain_i)_0^{n-1}$, there are $n$ number of edges with a flow of $1$ in the outgoing edges of $v_{src}$ in $G'$. Therefore, the amount of the constructed flow is $n$.
Since the volume of this generalized path list is no less than $X$, for the same reason as the sufficiency proof, the cost of the constructed flow is no greater than $-X$.
\end{proof}

\begin{example}\label{exp:list_flow}
This example continues Example \ref{exp:mini_cost} and illustrates the correspondence between a generalized path list in a DAG task and a flow in a flow network.
$(\chain_i)_0^1$ with $\chain_0=(v_0, v_3)$ and $\chain_1=(v_1)$ is a generalized path list for the DAG task $G$ in Fig. \ref{fig:ill_nested}. Its cardinality is 2 and its volume is 5.
By the necessity proof of Theorem \ref{thm:computation}, a corresponding flow for the flow network $G'$ in Fig. \ref{fig:mini_cost} can be constructed, which is: for each edge $(u, v)$ in $\{(v_{src}, v^{in}_0), (v^{in}_0, v^{out}_0), (v^{out}_0, v^{in}_3), (v^{in}_3, v^{out}_3), (v^{out}_3, v_{snk}), \\ (v_{src}, v^{in}_1), (v^{in}_1, v^{out}_1), (v^{out}_1, v_{snk})\}$, $f(u, v)=1$; for any other edge $(u, v)$ in $G'$, $f(u, v)=0$. The amount of this flow is 2 and the cost is -5.

Conversely, if we have a flow as above for $G'$, by the sufficiency proof of Theorem \ref{thm:computation}, a corresponding generalized path list for $G$ can be constructed, which is the same as $(\chain_i)_0^1$.
\end{example}

Theorem \ref{thm:computation} shows that the algorithms for the minimum-cost flow problem can be used to solve Problem \ref{pro:path}.
However, note that this is a subtle difference between the time complexities of the two problems.
As stated before, Problem \ref{pro:flow} has a time complexity $O(n(|E|+|V|log|V|))$, which is actually pseudo-polynomial. This is because the $n$ in Problem \ref{pro:flow} is a value and is not directly related to the size of the problem.
By Theorem \ref{thm:computation}, Problem~\ref{pro:path} also has a time complexity $O(n(|E|+|V|log|V|))$, which however is polynomial. This is because the $n$ in Problem~\ref{pro:path} is bounded by $|V|$, which is directly related to the size of Problem \ref{pro:path}.

\subsection{The Computation}
\label{sec:overall}
This subsection provides the overall algorithm to compute the proposed bound optimally as shown in Algorithm \ref{alg:computation}.

\begin{algorithm}[h]
    \caption{Optimal Computation of the Bound}\label{alg:computation}
    \DontPrintSemicolon
    \Input{the DAG task $G$, the number of cores $m$}
    \Output{the response time bound}
    $w \leftarrow$ the width of $G$ \\
    $n \leftarrow \min\{w, m\} $\\
    \ForEach{$j \leftarrow 0, \cdots, n-1$}{
        $W \leftarrow$ the maximum volume with cardinality $j+1$ \\
        $R_j \leftarrow len(G)+\frac{vol(G)-W}{m-j}$
    }
    \Return{$\min \limits_{j \in [0, n-1]} \{R_j\}$}
\end{algorithm}

In Algorithm \ref{alg:computation}, Line 1 computes the width (also called the degree of parallelism) of the DAG task. For algorithms of computing the width, see \cite{hopcroft1973n, garey1979computers, he2023degree}.
Line 4 computes the maximum volume of the generalized path lists with cardinality $j+1$ in $G$.
Line 4 is realized by first transforming the DAG task $G$ into a flow network $G'$ and second solving the minimum-cost flow problem for $G'$.
Line 5 computes a response time bound for $G$ on $m$ cores by using (\ref{equ:dag_bound}).
Since each $R_j$ is a safe response time bound by Lemma \ref{lem:dag_bound}, Line 7 takes the minimum one among all $R_j$ as the final response time bound for DAG task $G$ executing on $m$ cores under a work-conserving scheduler.

\textbf{Complexity.} The time complexity of Line 4 is $O(w(|E|+|V|log|V|))$, where $w$ is the width of $G$. The loop of Lines 3-6 executes for at most $w$ times. Therefore, the time complexity of Algorithm \ref{alg:computation} is $O(w^2(|E|+|V|log|V|))$.
We observe that the algorithm of Problem \ref{pro:flow} computes the minimum cost for flows of amount $n$ by iteratively computing the minimum cost for flows of amount $1, \cdots, n-1, n$. Therefore, with a well-integration with the algorithm of the minimum-cost flow problem in Line 4, Algorithm \ref{alg:computation} can be easily implemented with time complexity $O(w(|E|+|V|log|V|))$, where $w$ is the width of the DAG task and is bounded by the number of vertices~$|V|$.

\section{The Dominance Among Multi-Path Bounds}
\label{sec:dominance}
This section establishes the dominance among multi-path bounds: our bound dominates all three existing multi-path bounds \cite{he2022bounding, he2023degree, ueter2023parallel} and Graham's bound \cite{graham1969bounds}; the three existing multi-path bounds do not dominate each other.
For dominance, superior cases are provided; for nondominance, both superior cases and inferior cases are provided.

\textbf{Multi-Path Bound.} Regarding response time bound of a DAG task, the classic result, i.e. Graham's bound, utilizes the volume and the longest path of the DAG task for bounding response times. In contrast, results, such as the proposed bound, utilize the volume and multiple paths of the DAG task for bounding response times.
We call response time bounds using multiple paths of the DAG task as \emph{multi-path bounds}. Existing multi-path bounds include \cite{he2022bounding, he2023degree, ueter2023parallel} and the proposed bound in this paper.

We restate the response time bound of \cite{he2022bounding} as follows.
\begin{theorem}[Adapted from Theorem 2 of \cite{he2022bounding}]
\label{thm:rtss_bound}
Given a generalized path list $(\mypath_i)_0^k$ ($k \in [0, m-1]$) with $\mypath_0$ being the longest path of $G$, the response time $R$ of DAG $G$ scheduled by work-conserving scheduling on $m$ cores is bounded by:
\begin{equation}\label{equ:rtss_bound}
R\le \min \limits_{j \in [0, k]} \left\{ len(G)+\frac{vol(G)-\sum_{i=0}^{j} len(\mypath_i)}{m-j} \right\}
\end{equation}
\end{theorem}

\begin{theorem}\label{thm:our_rtss}
The bound in Theorem \ref{thm:our_bound} dominates the bound of \cite{he2022bounding} in Theorem \ref{thm:rtss_bound}.
\end{theorem}
\begin{proof}
Regarding the generalized path list $(\mypath_i)_0^k$ used to compute the two bounds, Theorem \ref{thm:rtss_bound} requires that the first generalized path $\mypath_0$ must be the longest path of $G$. However, Theorem \ref{thm:our_bound} does not have this requirement. This means that a generalized path list for Theorem \ref{thm:rtss_bound} is a generalized path list for Theorem \ref{thm:our_bound}. But the opposite is not true: a generalized path list for Theorem \ref{thm:our_bound} may not be a generalized path list for Theorem \ref{thm:rtss_bound}.
Note that (\ref{equ:our_bound}) is the same as (\ref{equ:rtss_bound}). The theorem follows.
\end{proof}

\begin{superior*}
The example in Fig. \ref{fig:ill_nested} of Section \ref{sec:problem} explains how our bound in Theorem \ref{thm:our_bound} can be smaller than the bound in Theorem \ref{thm:rtss_bound}.
For the DAG task in Fig. \ref{fig:ill_nested} scheduled on $m=2$ cores, our bound is $4$ and the bound of \cite{he2022bounding} is $5$.
\end{superior*}

\begin{theorem}[Adapted from Theorem 4 of \cite{he2023degree}]
\label{thm:date_bound}
Given a generalized path list $(\mypath_i)_0^k$ ($k \in [0, m-1]$), the response time $R$ of DAG $G$ scheduled by work-conserving scheduling on $m$ cores is bounded by:
\begin{equation}\label{equ:date_bound}
R\le len(G)+vol(G)-\sum_{i=0}^{j} len(\mypath_i)
\end{equation}
\end{theorem}

\begin{theorem}\label{thm:our_date}
The bound in Theorem \ref{thm:our_bound} dominates the bound of \cite{he2023degree} in Theorem \ref{thm:date_bound}.
\end{theorem}
\begin{proof}
For the two bounds, the generalized path lists $(\mypath_i)_0^k$ are the same: a generalized path list for Theorem \ref{thm:date_bound} is a generalized path list for Theorem \ref{thm:our_bound}, and vise versa.
For an arbitrary generalized path lists $(\mypath_i)_0^k$ ($k \in [0, m-1]$), we have
\begin{multline*}
len(G)+\frac{vol(G)-\sum_{i=0}^{j} len(\mypath_i)}{m-j} \\ \le len(G)+vol(G)-\sum_{i=0}^{j} len(\mypath_i)
\end{multline*}
The theorem follows.
\end{proof}

\begin{figure}[t]
\centering
\subfloat[]{
    \includegraphics[width=0.28\linewidth]{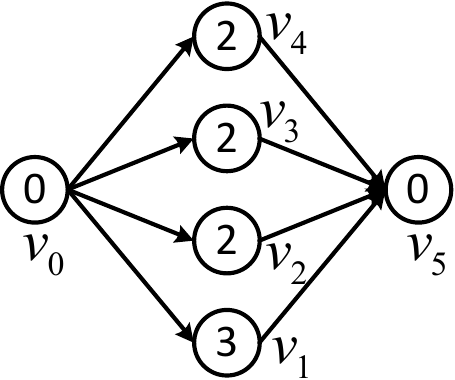}
    \label{fig:dag_date}
}
\hfil
\subfloat[]{
    \includegraphics[width=0.33\linewidth]{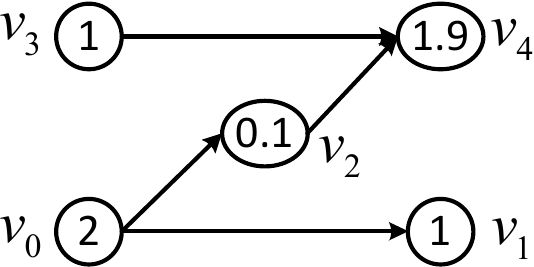}
    \label{fig:dag_chen}
}
\caption{(a) A DAG example for Theorem \ref{thm:our_date}. (b) A DAG example for Theorem~\ref{thm:our_chen}.}
\label{fig:domi_example}
\end{figure}

\begin{superior*}
Fig. \ref{fig:dag_date} shows a DAG task $G$ where $v_0$ and $v_5$ are dummy source vertex and dummy sink vertex, respectively. The length $len(G)=3$ and the volume $vol(G)=9$. Let the number of cores $m=2$.
$(\chain_i)_0^1$ with $\chain_0=(v_0, v_1, v_5)$ and $\chain_1=(v_2)$ is a generalized path list.
The bound in (\ref{equ:our_bound}) is computed as
\begin{align*}
R   &\le \min\{3+(9-3)/2, 3+(9-3-2)/(2-1)\} \\
    &=\min\{6, 7\}=6
\end{align*}
The bound in (\ref{equ:date_bound}) is computed as
$$
R \le 3+9-5=7
$$
For the DAG task in Fig. \ref{fig:dag_date} scheduled on $m=2$ cores, our bound is smaller than the bound of \cite{he2023degree}.
\end{superior*}

\begin{theorem}[Adapted from Theorem 6 of \cite{ueter2023parallel}]
\label{thm:chen_bound}
Given a generalized path list $(\mypath_i)_0^k$ ($k \in [0, m-1]$), the response time $R$ of DAG $G$ scheduled by work-conserving scheduling with preemptive vertex-level priorities on $m$ cores is bounded by:
\begin{equation}\label{equ:chen_bound}
R\le \min \limits_{j \in [0, k]} \left\{ len(G)+\frac{vol(G)-\sum_{i=0}^{j} len(\mypath_i)}{m-j} \right\}
\end{equation}
\end{theorem}

The scheduling algorithm studied in \cite{ueter2023parallel} is a restricted version of work-conserving scheduling. It is work-conserving scheduling, but with more constraints than the scheduling algorithm studied in \cite{graham1969bounds, he2022bounding, he2023degree} and this paper. This paper and \cite{graham1969bounds, he2022bounding, he2023degree} only assume that the scheduling algorithm satisfies the work-conserving property. \cite{ueter2023parallel} assumes the work-conserving property and further requires vertex-level priorities to restrict the execution behavior of vertices.

We remark that a response time bound is an upper bound on a set of response times of all the execution sequences of a DAG task.
For a DAG task, the set of execution sequences is subject to the scheduling algorithm.
Let $\mathcal{S}$ denote a scheduling algorithm. Let $\mathcal{S}(G, m)$ denote the set of execution sequences of a DAG task $G$ scheduled by $\mathcal{S}$ on $m$ cores. Let $R(\eseq)$ denote the response time of an execution sequence $\eseq$.
Formally, a response time bound is an upper bound on the following set.
\begin{equation}\label{equ:boundset}
\{R(\eseq) \mid \eseq \in \mathcal{S}(G, m) \}
\end{equation}

\begin{theorem}\label{thm:our_chen}
The bound in Theorem \ref{thm:our_bound} dominates the bound of \cite{ueter2023parallel} in Theorem \ref{thm:chen_bound}.
\end{theorem}
\begin{proof}
On the one hand, since (\ref{equ:our_bound}) is the same as (\ref{equ:chen_bound}) and both bounds do not have the constraint of the longest path (i.e. without requiring that the first generalized path in the generalized path list should be the longest path),
and since our algorithm for computing the generalized path list in Section \ref{sec:optimal} is optimal, our bound is no larger than the bound of \cite{ueter2023parallel}.

On the other hand, the scheduling algorithm in \cite{ueter2023parallel} (denoted as $\mathcal{S}_1$) is a restricted version of the scheduling algorithm in this paper (denoted as $\mathcal{S}_0$).
For a DAG task $G$ scheduled on $m$ cores, we have
$\mathcal{S}_1(G, m) \subseteq \mathcal{S}_0(G, m)$.
Therefore,
\begin{equation}\label{equ:set_domi}
\{R(\eseq) \mid \eseq \in \mathcal{S}_1(G, m) \} \subseteq \{R(\eseq) \mid \eseq \in \mathcal{S}_0(G, m) \}
\end{equation}

In summary, we have less constraints on the scheduling algorithm, thus deriving a bound on a larger set of response times (see Equation \ref{equ:set_domi}).
What's more, our bound is no larger than the bound of \cite{ueter2023parallel}.
Therefore, our bound dominates the bound of \cite{ueter2023parallel}.
\end{proof}

\begin{superior*}
Fig. \ref{fig:dag_chen} shows a DAG task $G$. The length $len(G)=4$ and the volume $vol(G)=6$. Let the number of cores $m=2$.
Using the algorithm in Section \ref{sec:optimal}, a generalized path list $(\chain_i)_0^2$ with $\chain_0=(v_0, v_1)$,  $\chain_1=(v_3, v_4)$ and $\chain_2=(v_2)$ is computed. With $(\chain_i)_0^2$, the bound in (\ref{equ:our_bound}) is computed as
\begin{align*}
R   &\le \min\{4+(6-3)/2, 4+(6-3-2.9)/(2-1)\} \\
    &=\min\{5.5, 4.1\}=4.1
\end{align*}
Using Algorithm 1 of \cite{ueter2023parallel}, a generalized path list $(\chain_i)_0^2$ with $\chain_0=(v_0, v_2, v_4)$,  $\chain_1=(v_1)$ and $\chain_2=(v_3)$ is computed. With $(\chain_i)_0^2$, the bound in (\ref{equ:chen_bound}) is computed as
\begin{align*}
R   &\le \min\{4+(6-4)/2, 4+(6-4-1)/(2-1)\} \\
    &=\min\{5, 5\}=5
\end{align*}
For the DAG task in Fig. \ref{fig:dag_chen} scheduled on $m=2$ cores, our bound is smaller than the bound of \cite{ueter2023parallel}.
\end{superior*}

Although the bound of \cite{ueter2023parallel} is no larger than bounds of \cite{graham1969bounds, he2022bounding, he2023degree}, since \cite{ueter2023parallel} has more constraints on the scheduling algorithm to restrict the execution behavior of vertices, it is inappropriate to say that the bound of \cite{ueter2023parallel} dominates bounds of \cite{graham1969bounds, he2022bounding, he2023degree}.

\begin{theorem}\label{thm:rtss_date}
There is no dominance between the bound of \cite{he2022bounding} in Theorem \ref{thm:rtss_bound} and the bound of \cite{he2023degree} in Theorem \ref{thm:date_bound}.
\end{theorem}

\begin{superior*}
The superior case of Theorem \ref{thm:our_date} is also a superior case of Theorem \ref{thm:rtss_date}.
In Fig. \ref{fig:dag_date}, for generalized path list $(\chain_i)_0^1$ with $\chain_0=(v_0, v_1, v_5)$ and $\chain_1=(v_2)$, $\chain_0$ is the longest path.
So, the bound in (\ref{equ:rtss_bound}) is 6 and the bound in (\ref{equ:date_bound}) is 7.
Therefore, for the DAG task in Fig. \ref{fig:dag_date} scheduled on $m=2$ cores, the bound of \cite{he2022bounding} is smaller than the bound of \cite{he2023degree}.
\end{superior*}

\begin{inferior*}
The example in Fig. \ref{fig:ill_nested} of Section \ref{sec:problem} explains how the bound of \cite{he2022bounding} can be larger than the bound of \cite{he2023degree}.
For generalized path list $(\chain_i)_0^1$ with $\chain_0=(v_0, v_3)$ and $\chain_1=(v_1)$, $\chain_0$ is the longest path.
So, the bound in (\ref{equ:rtss_bound}) is 5.
Note that the bound of \cite{he2023degree} does not have the constraint of the longest path. So we can use a different generalized path list to compute the bound of \cite{he2023degree}.
For generalized path list $(\chain_i)_0^1$ with $\chain_0=(v_0, v_1)$ and $\chain_1=(v_2, v_3)$,
the bound in (\ref{equ:date_bound}) is computed as
$$
R \le 4+6-6=4
$$
Therefore, for the DAG task in Fig. \ref{fig:ill_nested} scheduled on $m=2$ cores, the bound of \cite{he2022bounding} is larger than the bound of \cite{he2023degree}.
\end{inferior*}

In addition, \cite{he2022bounding} shows that the bound of \cite{he2022bounding} dominates Graham's bound. By Theorem \ref{thm:our_rtss}, our bound also dominates Graham's bound.
Moreover, \cite{he2023degree} shows that there is no dominance between the bound of \cite{he2023degree} and Graham's bound.
In summary, we provide an overall picture in Fig. \ref{fig:dominance} concerning the dominance and nondominance among the three existing multi-path bounds, the proposed multi-path bound and the classic Graham's bound.

\begin{figure}[t]
  \centering
  \includegraphics[width=0.75\linewidth]{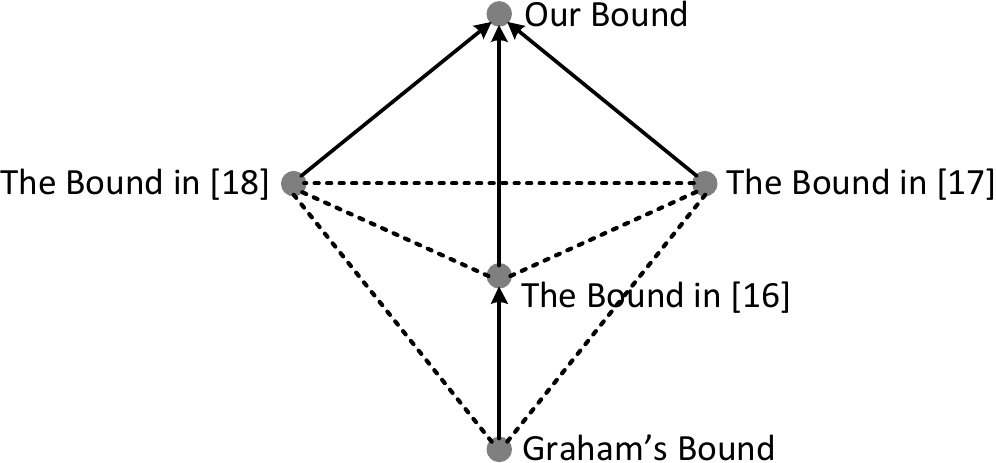}\\
  \caption{A hierarchy of multi-path bounds. Solid lines with arrows indicate dominance. Dashed lines indicate nondominance.}
  \label{fig:dominance}
\end{figure}

\section{The Sustainability of Multi-Path Bounds}
\label{sec:sustain}
This section analyzes the sustainability of multi-path bounds.
Intuitively, sustainability means that a response time bound is still safe when the system parameters get better (for example, the WCET of vertices decreases or the number of cores increases).
Obviously, for the four multi-path bounds (i.e. \cite{he2022bounding, he2023degree, ueter2023parallel} and our bound), when the number of cores increases, these response time bounds will not increase, thus still being safe bounds.
In the following, this paper focuses on sustainability with respect to the WCET of vertices in the DAG task.
We distinguish two types of sustainability.

\begin{definition}[Sustainability \cite{baruah2006sustainable}]\label{def:sustain}
A response time bound is sustainable if the bound holds true when the WCET of vertices decreases.
\end{definition}

\begin{definition}[Self-Sustainability \cite{baker2009sustainable}]\label{def:self}
A response time bound is self-sustainable if the bound does not increase when the WCET of vertices decreases.
\end{definition}

The above two definitions are adapted from respective literature regarding to our problem setting.
At first glance, it seems that the two definitions are the same. However, they are not: self-sustainability is a stronger notion than sustainability. Next, we introduce some notations to illustrate the difference formally.

For a DAG task $G$ and the number of cores $m$, a response time bound $B(G, m)$ can be computed according to a specific method under a designated scheduling algorithm. Here, function $B(\boldsymbol{\cdotp}, \boldsymbol{\cdot})$ denotes a specific computing method and a designated scheduling algorithm.
For example, this paper proposes a method for computing the response time bound under work-conserving scheduling in Section \ref{sec:bound} and Section \ref{sec:computation}.
For the DAG task $G=(V, E)$, when the WCET of vertices decreases, we have a new DAG task $G'=(V, E)$. $G'$ has the same vertex set and edge set as $G$. But $G'$ has different WCETs for vertices $c'(v)$ satisfying $\forall v \in V, c'(v) \le c(v)$.

If $B(G, m)$ holds true for $G'$, i.e., $B(G, m)$ is still an upper bound on the response times of $G'$ executing on $m$ cores under the designated scheduling algorithm, we say that $B(G, m)$ is \emph{sustainable}.
If the bound does not increase, i.e. $B(G, m) \ge B(G', m)$, we say that $B(G, m)$ is \emph{self-sustainable}.

By the definition of response time bound, $B(G', m)$ is an upper bound on the response times of $G'$ executing on $m$ cores under the designated scheduling algorithm. If $B(G, m) \ge B(G', m)$, it means that $B(G, m)$ is also an upper bound on the response times of $G'$ executing on $m$ cores under the designated scheduling algorithm. Therefore, self-sustainability implies sustainability. But the opposite is not true: sustainability does not imply self-sustainability.

In the context of real-time scheduling, sustainability is critical for the correctness of a response time bound. Without sustainability, a response time bound cannot even be called to be correct. Many response time bounds are sustainable. For example, the multi-path bound in \cite{ueter2023parallel} is proved to be sustainable in Corollary 7 of \cite{ueter2023parallel}.
For our bound in Theorem~\ref{thm:our_bound}, since we inherently take into consideration the behavior where vertices may execute for less than its WCET, the correctness of the proof in Section \ref{sec:bound} directly means sustainability.
We mention that all four multi-path bounds and Graham's bound are sustainable.

However, it is still unknown whether these multi-path bounds are self-sustainable or not.
As stated in \cite{baker2009sustainable}, self-sustainability is important in the incremental and interactive design process, which is typically used in the design of real-time systems and in the evolutionary development of fielded systems, such as the real-time system in \cite{yi2022mimos}. In this section, we try to analyze the self-sustainability of multi-path bounds.

\begin{theorem}\label{thm:rtss_sustain}
The bound in \cite{he2022bounding} (Theorem \ref{thm:rtss_bound} and Algorithm~2 of \cite{he2022bounding}) is not self-sustainable.
\end{theorem}
\begin{proof}
By Definition \ref{def:self}, to prove that a bound is not self-sustainable, it is sufficient to construct a counter-example.
Let $G$ be the DAG task in Fig. \ref{fig:ill_nested} but with $c(v_1)=2.1$. When the WCET of vertices in $G$ decreases, let $G'$ be exactly the DAG task in Fig. \ref{fig:ill_nested}. Let the number of cores $m=2$.

For $G$, using Algorithm 2 of \cite{he2022bounding}, a generalized path list $(\chain_i)_0^1$ with $\chain_0=(v_0, v_1)$ and $\chain_1=(v_2, v_3)$ is computed. With $(\chain_i)_0^1$, by (\ref{equ:rtss_bound}), a response time bound $B(G, m)$ is computed.
\begin{align*}
B(G, m) &=\min\{4.1+(7.1-4.1)/2, \\ & \qquad \qquad \qquad \qquad 4.1+(7.1-4.1-3)/(2-1)\} \\
        &=\min\{5.6, 4.1\}=4.1
\end{align*}

For $G'$, using Algorithm 2 of \cite{he2022bounding}, a generalized path list $(\chain_i)_0^2$ with $\chain_0=(v_0, v_3)$,  $\chain_1=(v_1)$ and $\chain_2=(v_2)$ is computed. With $(\chain_i)_0^2$, by (\ref{equ:rtss_bound}), a response time bound $B(G', m)$ is computed.
\begin{align*}
B(G', m) &=\min\{4+(6-4)/2, 4+(6-4-1)/(2-1)\} \\
        &=\min\{5, 5\}=5
\end{align*}

When the WCET of vertices in $G$ decreases, we have $B(G, m) < B(G', m)$, which means that the bound increases.
Therefore, the bound in \cite{he2022bounding} is not self-sustainable.
\end{proof}

\begin{theorem}\label{thm:chen_sustain}
The bound in \cite{ueter2023parallel} (Theorem \ref{thm:chen_bound} and Algorithm~1 of \cite{ueter2023parallel}) is not self-sustainable.
\end{theorem}
\begin{proof}
We also use a counter-example to prove this theorem.
Let $G$ be the DAG task in Fig. \ref{fig:dag_chen} but with $c(v_1)=2.1$. When the WCET of vertices in $G$ decreases, let $G'$ be exactly the DAG task in Fig. \ref{fig:dag_chen}. Let the number of cores $m=2$.

For $G$, using Algorithm 1 of \cite{ueter2023parallel}, a generalized path list $(\chain_i)_0^2$ with $\chain_0=(v_0, v_1)$,  $\chain_1=(v_3, v_4)$ and $\chain_2=(v_2)$ is computed. With $(\chain_i)_0^2$, by (\ref{equ:chen_bound}), a response time bound $B(G, m)$ is computed.
\begin{align*}
B(G, m) &=\min\{4.1+(7.1-4.1)/2, \\ & \qquad \qquad \quad \qquad 4.1+(7.1-4.1-2.9)/(2-1)\} \\
        &=\min\{5.6, 4.2\}=4.2
\end{align*}

For $G'$, using Algorithm 1 of \cite{ueter2023parallel}, a generalized path list $(\chain_i)_0^2$ with $\chain_0=(v_0, v_2, v_4)$,  $\chain_1=(v_1)$ and $\chain_2=(v_3)$ is computed. With $(\chain_i)_0^2$, by (\ref{equ:chen_bound}), a response time bound $B(G', m)$ is computed.
\begin{align*}
B(G', m) &=\min\{4+(6-4)/2, 4+(6-4-1)/(2-1)\} \\
        &=\min\{5, 5\}=5
\end{align*}

When the WCET of vertices in $G$ decreases, we have $B(G, m) < B(G', m)$, which means that the bound increases.
Therefore, the bound in \cite{ueter2023parallel} is not self-sustainable.
\end{proof}

Note that in the above two proofs, after using the same computing method for $G$ and $G'$, although the response time bound of $G'$ (i.e. $B(G', m)$) is larger than that of $G$ (i.e. $B(G, m)$), $B(G, m)$ is still a safe response time bound for $G'$, since this is implied by the sustainability (Definition \ref{def:sustain}) of the two bounds in \cite{he2022bounding, ueter2023parallel}.
The reason that leads to $B(G, m) < B(G', m)$ lies in the bound-computation method, not in the scheduling algorithm itself.
In other words, there is some sort of ``unsustainability'' within the bound-computation methods of \cite{he2022bounding, ueter2023parallel}.

\begin{theorem}\label{thm:our_sustain}
The proposed response time bound in Theorem~\ref{thm:our_bound} and Algorithm~\ref{alg:computation} is self-sustainable.
\end{theorem}
\begin{proof}
Let $G=(V, E)$ denote an arbitrary DAG task and $G'=(V, E)$ denote the DAG task when the WCET of vertices in $G$ decreases. We have $len(G) \ge len(G')$.

Recall that in Section \ref{sec:dag}, for an arbitrary vertex set $U \subseteq V$, $vol(U)$ is defined to be $\sum_{v\in U} c(v)$ in task $G$.
Now, we have two DAG tasks $G$ and $G'$ with the same vertex set $V$ but with different WCETs $c(v)$ and $c'(v)$.
So we define $vol'(U) \coloneqq \sum_{v\in U} c'(v)$. Function $vol'(\boldsymbol{\cdot})$ is introduced to denote the volume of a vertex set for $G'$.

Let $w$ denote the width of $G$. Using the method in Section \ref{sec:optimal}, $\forall j \in [0, w-1]$, we can compute a generalized path list $(\chain_i)_0^j$ with the maximum volume in $G$, and a generalized path list $(\chain'_i)_0^j$ with the maximum volume in $G'$.
Let $V_0$ denote the vertex set including vertices in $(\chain_i)_0^j$ and $V_1$ denote the vertex set including vertices in $(\chain'_i)_0^j$.

With $len(G) \ge len(G')$, by (\ref{equ:dag_bound}) and Line 5 of Algorithm~\ref{alg:computation}, to prove the self-sustainability, it is sufficient to prove that (\ref{equ:sustain}) holds.
\begin{equation}\label{equ:sustain}
vol(V)-vol(V_0) \ge vol'(V)-vol'(V_1)
\end{equation}
We define $V^C \coloneqq V \setminus (V_0 \cup V_1)$, $V^I \coloneqq V_0 \cap V_1$, $V^D_0 \coloneqq V_0 \setminus V^I$, $V^D_1 \coloneqq V_1 \setminus V^I$. These four vertex sets are mutually disjoint and we have $V = V^C \cup V^I \cup V^D_0 \cup V^D_1$.

For cardinality $j+1$, $(\chain'_i)_0^j$ is the generalized path list with the maximum volume in $G'$. We have
\begin{align*}
vol'(V_0) &\le vol'(V_1) \\
\Longrightarrow vol'(V^D_0)+vol'(V^I) &\le vol'(V^D_1)+ vol'(V^I) \\
\Longrightarrow vol'(V^D_0) &\le vol'(V^D_1)
\end{align*}
Compared to $G$, the WCETs of some vertices decrease in $G'$. Therefore, $vol'(V^D_1) \le vol(V^D_1)$, which means that $vol'(V^D_0) \le vol(V^D_1)$. We have
\begin{align*}
 vol(V^D_1) &\ge vol'(V^D_0) \\
\Longrightarrow vol(V^D_1)+vol(V^C) &\ge vol'(V^D_0)+ vol'(V^C) \\
\Longrightarrow vol(V^D_1 \cup V^C) &\ge vol'(V^D_0 \cup V^C) \\
\Longrightarrow vol(V \setminus V_0) &\ge vol'(V \setminus V_1) \\
\Longrightarrow vol(V)-vol(V_0) &\ge vol'(V)-vol'(V_1)
\end{align*}
which is (\ref{equ:sustain}). The theorem is proved.
\end{proof}

It is unknown whether the bound in \cite{he2023degree} is self-sustainable or not.
Table \ref{tab:sustain} summarizes the sustainability and self-sustainability of multi-path bounds and Graham's bound.
Our bound is the only multi-path bound that is proved to be self-sustainable.

\begin{table}[t]
\centering
\caption{The Sustainability of Multi-Path Bounds}
\label{tab:sustain}
\begin{tabular}{lcc}
\hline
                            &\textbf{Sustainability}    &\textbf{Self-Sustainability} \\
\hline
Graham's bound                          &\checkmark         &\checkmark     \\
The bound in \cite{he2023degree}        &\checkmark         &\textsf{?}     \\
The bound in \cite{he2022bounding}      &\checkmark         &\ding{53}      \\
The bound in \cite{ueter2023parallel}   &\checkmark         &\ding{53}      \\
Our bound                               &\checkmark         &\checkmark     \\
\hline
\end{tabular}
\end{table}

%\section{Multi-Path Bound for Other Task Models}
%\label{sec:conditional}
%\input{conditional}

\section{Evaluation}
\label{sec:evaluation}
This section evaluates the performance of the proposed method using randomly generated DAG tasks.
Section \ref{sec:eval_bound} compares the proposed response time bound with other multi-path bounds.
Section \ref{sec:eval_sched} evaluates the schedulability of task sets for our bound applied in federated scheduling and other state-of-the-art methods of scheduling DAG tasks.

\subsection{Response Time Bound of DAG Tasks}
\label{sec:eval_bound}

\begin{figure*}[t]
\centering
\subfloat[]{
    \includegraphics[width=0.25\linewidth]{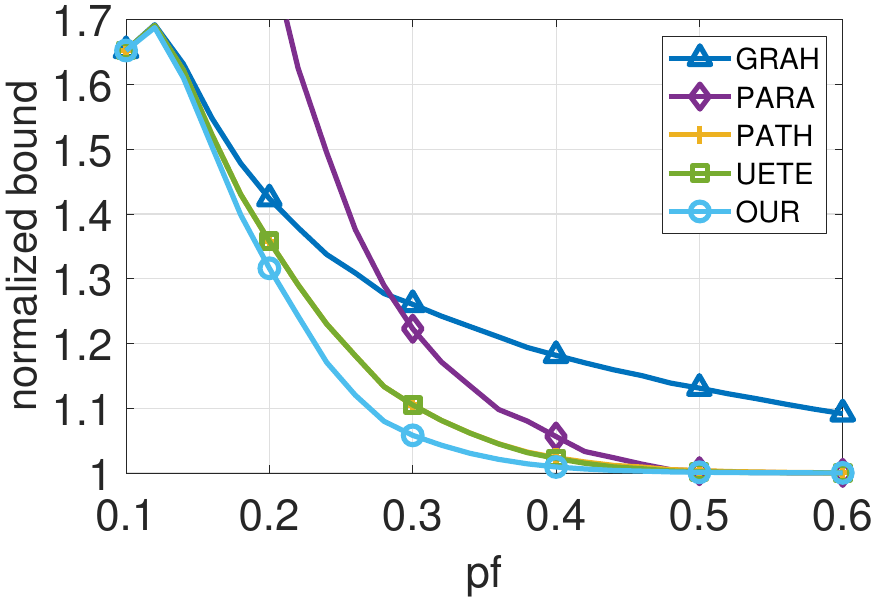}
    \label{fig:p_4}
}
\hfil
\subfloat[]{
    \includegraphics[width=0.25\linewidth]{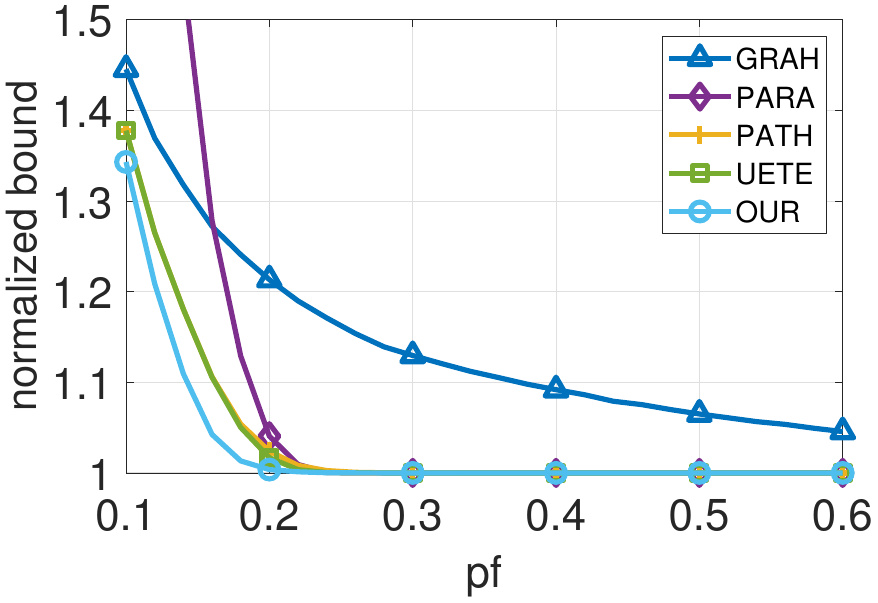}
    \label{fig:p_8}
}
\hfil
\subfloat[]{
    \includegraphics[width=0.25\linewidth]{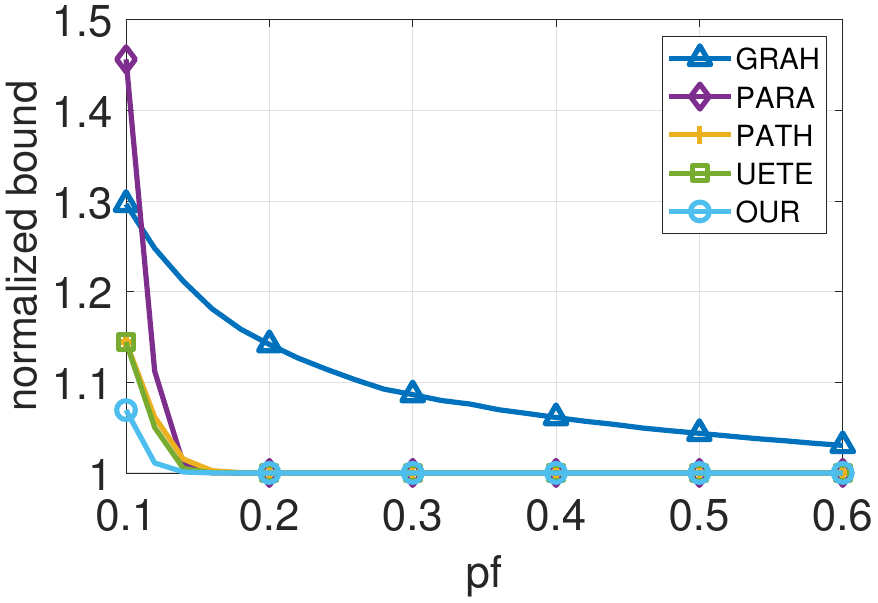}
    \label{fig:p_12}
}
\caption{{\diff Normalized bound with changing the parallelism factor. (a) $m=4$. (b) $m=8$. (c) $m=12$.}}
\label{fig:eval_bound}
\end{figure*}

\begin{figure*}[t]
\centering
\subfloat[]{
    \includegraphics[width=0.205\linewidth]{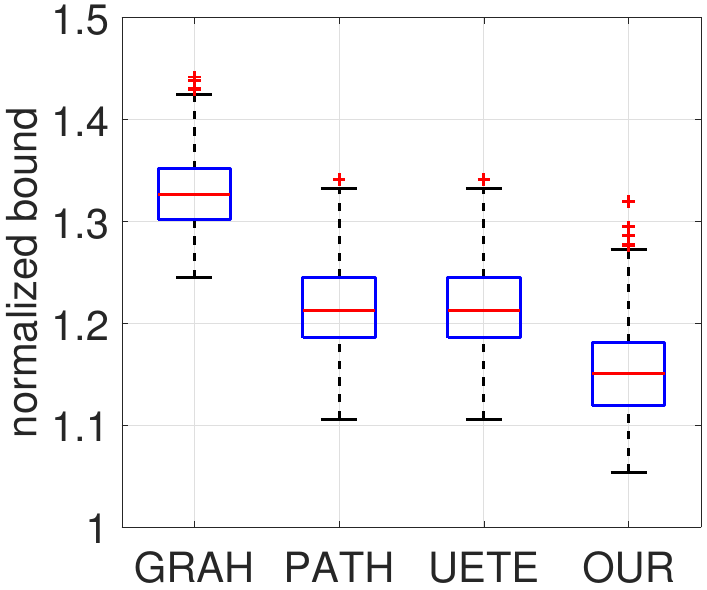}
    \label{fig:b_4}
}
\hfil
\subfloat[]{
    \includegraphics[width=0.205\linewidth]{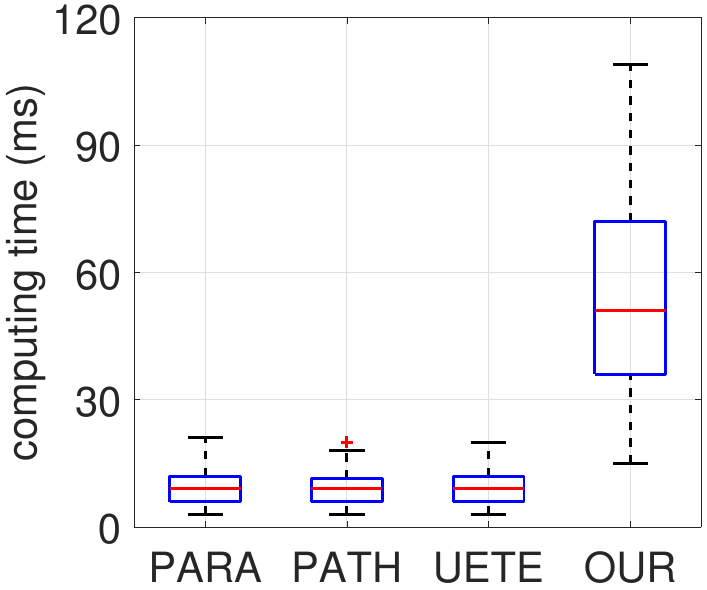}
    \label{fig:t_4}
}
\hfil
\subfloat[]{
    \includegraphics[width=0.25\linewidth]{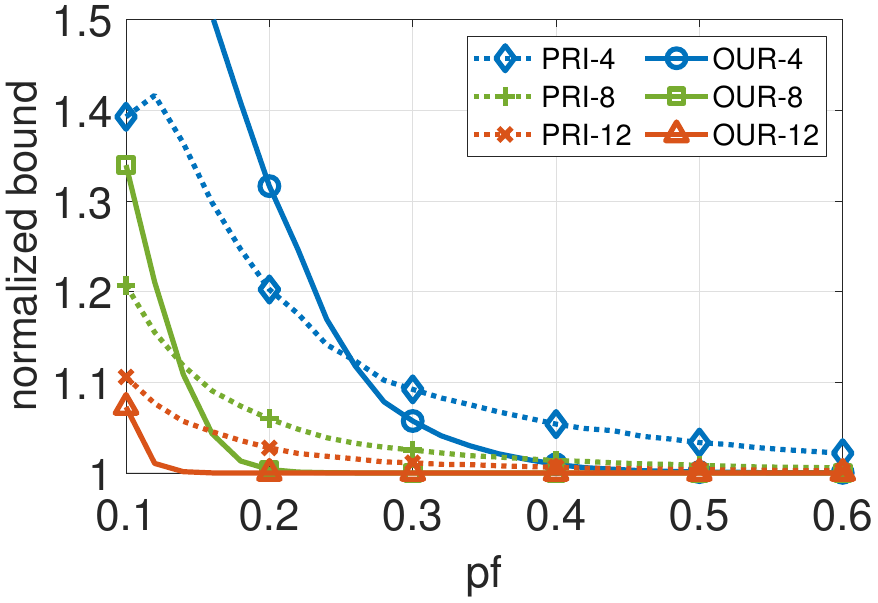}
    \label{fig:p_pr}
}
\hfil
\subfloat[]{
    \includegraphics[width=0.25\linewidth]{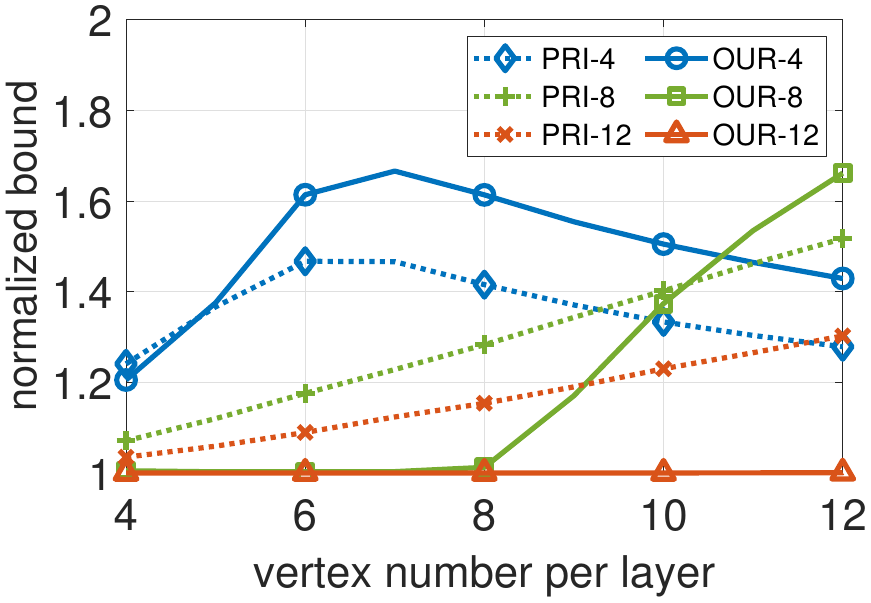}
    \label{fig:p_ly}
}
\caption{{\diff (a) The variation of normalized bounds. $m=4$, $\mathit{pf} \in [0.2, 0.3]$. (b) The computing time of different bounds. $m=4$, $\mathit{pf} \in [0.2, 0.3]$. (c) Compare our multi-path bound with the priority-based bound when $m=4, 8, 12$. PRI-4 means the bound PRI on the number of cores $m=4$. (d) Comparison using the layer-by-layer method to generate DAGs. $m=4, 8, 12$; the number of layers is in $[5, 15]$; the probability of adding edges between layers is 0.5; the WCET of vertices $c(v) \in [5, 100]$.}}
\label{fig:eval_other}
\end{figure*}

This subsection compares the following response time bounds of a DAG task.
\begin{itemize}
    \item \textsf{GRAH}. The classic result Graham's bound in \cite{graham1969bounds}.
    {\diff \item \textsf{PRI}.  The priority-based bound in \cite{zhao2020dag} that does not utilize the multi-path technique.}
    \item \textsf{PARA}. The multi-path bound in \cite{he2023degree}: Theorem \ref{thm:date_bound} and Algorithm 3 of \cite{he2023degree}.
    \item \textsf{PATH}. The multi-path bound in \cite{he2022bounding}: Theorem \ref{thm:rtss_bound} and Algorithm~2 of \cite{he2022bounding}.
    \item \textsf{UETE}. The multi-path bound in \cite{ueter2023parallel}: Theorem \ref{thm:chen_bound} and Algorithm~1 of \cite{ueter2023parallel}.
    \item \textsf{OUR}. The proposed multi-path bound: Theorem \ref{thm:our_bound} and Algorithm \ref{alg:computation}.
\end{itemize}
Note that the scheduling algorithms for \textsf{GRAH}, \textsf{PARA}, \textsf{PATH}, \textsf{OUR} are the same and only assume the work-conserving property. However, the scheduling algorithms for \textsf{PRI} and \textsf{UETE} are more complex and rely on vertex-level priorities to control the execution behavior of vertices in the DAG task.
\textsf{PARA}, \textsf{PATH}, \textsf{UETE} are all state-of-the-art results regarding the response time bound of a DAG task. Prior to this work, these three bounds have not been compared to each other.
In the evaluation, all bounds are normalized to a theoretical lower bound $\max\{len(G), \frac{vol(G)}{m}\}$. No response time bound for a DAG task $G$ scheduled on $m$ cores can be less than this lower bound.

\noindent
\textbf{Task Generation.}
The DAG tasks are generated using the Erd\"os-R\'enyi method \cite{cordeiro2010random}.
First, the number of vertices $|V|$ is randomly chosen in $[150, 250]$.
Second, for each pair of vertices $v_i, v_j$ and $i<j$, it generates a random value in $[0, 1]$.
If this generated value is less than a predefined \emph{parallelism factor} $\mathit{pf}$, an edge $(v_i, v_j)$ is added to the graph.
The larger $\mathit{pf}$, which means that there are more edges, the more sequential the graph is.
When adding edges, we ensure $i<j$ to avoid loops in the generated graph.
Now the topology of the graph (i.e. an adjacency matrix) is generated.
{ \diff
If the generated graph has multiple source or sink vertices, a vertex with zero WCET is added to ensure that the graph has single source or sink.
The WCET of each vertex $c(v)$ is randomly chosen in $[5, 100]$.}

Fig. \ref{fig:eval_bound} reports the normalized bound of different methods with changing the number of cores $m$ and the parallelism factor $\mathit{pf}$. For each data point in Fig.~\ref{fig:eval_bound}, we randomly generate 500 tasks to compute the average normalized bound.
Since all bounds are normalized to the theoretical lower bound as stated before, no normalized bounds can be less than 1 in the figure.
When $\mathit{pf}$ is small, the DAG task is highly parallel and has a large degree of parallelism. Since \textsf{PARA} heavily relies on the degree of parallelism of the DAG task, \textsf{PARA} becomes quite large when $\mathit{pf}$ is small. Therefore, for small values of $\mathit{pf}$, the data of \textsf{PARA} are not completely presented in Fig. \ref{fig:eval_bound}.
\textsf{UETE} is almost the same as (only slightly better than) \textsf{PATH}. This is because (\ref{equ:chen_bound}) in Theorem \ref{thm:chen_bound} is exactly the same as (\ref{equ:rtss_bound}) in Theorem \ref{thm:rtss_bound}. The only difference regarding the computation of the two bounds lies in the algorithms for computing the generalized path list: the algorithm for \textsf{UETE} (i.e. Algorithm~1 of \cite{ueter2023parallel}) leverages the degree of parallelism, while the algorithm for \textsf{PATH} (i.e. Algorithm~2 of \cite{he2022bounding}) does not leverage it.
\textsf{OUR}, through eliminating the constraint of the longest path and thus enabling the optimal computation of generalized path list, has the best performance, pushing the limits of multi-path bounds.
For most parts of Fig. \ref{fig:p_12} (such as $\mathit{pf} \ge 0.2$), our bound is the same as other multi-path bounds. This is simply because all multi-path bounds are equal to the theoretical lower bound. No methods can further reduce the response time bound in this case.
We can see that the data reported in Fig. \ref{fig:eval_bound} are consistent with theoretical analysis, i.e. the dominance among multi-path bounds in Section \ref{sec:dominance}.

In Fig. \ref{fig:eval_bound}, with the increase of $\mathit{pf}$, all bounds decrease and all multi-path bounds approach the theoretical lower bound (i.e. normalized bounds approach 1 in Fig. \ref{fig:eval_bound}). This is because when the parallelism factor $\mathit{pf}$ increases, there are more edges and the generated DAG becomes more sequential. So the degree of parallelism is more likely to be less than the number of cores. In this case, multi-path bounds will approach the length of the longest path, which is a theoretical lower bound.
This is also the reason why multi-path bounds approach 1 more quickly for a larger number of cores, since for a larger number of cores $m$, the degree of parallelism is also more likely to be less than $m$. We can discern this trend by comparing Figs.~\ref{fig:p_4}, \ref{fig:p_8}, \ref{fig:p_12}. For example, multi-path bounds are near to 1 for $\mathit{pf} \ge 0.5$ in Fig.~\ref{fig:p_4}; they are near to 1 for $\mathit{pf} \ge 0.3$ in Fig.~\ref{fig:p_8}.
In Fig.~\ref{fig:p_4}, for $\mathit{pf} \le 0.15$, the bounds temporarily increase.
This is because, for extremely high-parallel tasks (i.e. for $\mathit{pf}$ being very small), all paths in the task are very short.
In this case, all bounds, except for \textsf{PARA}, have a tendency of being close to $\frac{vol(G)}{m}$, which is another theoretical lower bound.
Note that for the computing equations, all other bounds have an item of volume being divided by $m$, such as the second item of (\ref{equ:our_bound}); but \textsf{PARA} does not have this kind of item in its computing equation, see (\ref{equ:date_bound}).

Since Fig. \ref{fig:eval_bound} only reports the average of normalized bounds, and since the performance of our bound is close to other multi-path bounds in some part of Fig. \ref{fig:eval_bound}, we select particular parts of Fig. \ref{fig:eval_bound} where \textsf{OUR} is more effective to examine the performances more closely and to show the variation of multi-path bounds.
The results are shown in Fig. \ref{fig:b_4}, where the number of cores $m=4$ and $\mathit{pf}$ is randomly chosen in $[0.2, 0.3]$.
For each box plot, 1000 DAG tasks are generated. As stated before, since \textsf{PARA} is too large in the selected settings, \textsf{PARA} is not included in Fig. \ref{fig:b_4}.
{\diff
Fig. \ref{fig:t_4} reports the time needed to compute the compared bounds. The setting of Fig. \ref{fig:t_4} is the same as Fig. \ref{fig:b_4}. Since the computation of \textsf{GRAH} needs less than 1 millisecond, \textsf{GRAH} is not included in Fig. \ref{fig:t_4}. The computing time of our bound is generally four times longer than other bounds.
Fig. \ref{fig:p_pr} shows the result of comparing our bound with the bound in \cite{zhao2020dag}, which utilizes vertex-level priority to reduce response time bound. The experiments are conducted on the number of cores $m=4, 8, 12$.
When $\mathit{pf}$ is small, meaning that the graph has more parallel vertices and fewer long paths, the priority-based bound \textsf{PRI} is more effective.
When $\mathit{pf}$ is large, meaning that the graph has fewer parallel vertices and more long paths, the multi-path bound \textsf{OUR} is more effective.
It also shows that when the number of cores is larger, the multi-path technique becomes more effective than the priority-based technique.
Fig. \ref{fig:p_ly} reports the result of comparing the same bounds as Fig. \ref{fig:p_pr}, but employing a different method to generate DAG tasks, i.e., the layer-by-layer method used in \cite{ueter2023parallel, zhao2020dag, zhao2022dag}. The observations obtained in Fig. \ref{fig:p_pr} can also be seen in Fig. \ref{fig:p_ly}. Note that when $\mathit{pf}$ in Fig. \ref{fig:p_pr} gets smaller and when the vertex number per layer in Fig. \ref{fig:p_ly} gets larger, the generated graph becomes more parallel.}
Same as \textsf{PARA} and \textsf{PATH}, experiments (not reported in the paper due to page limits) demonstrate that the performance of our bound is irrelevant to the number of vertices in the DAG task.

\subsection{Schedulability of DAG Task Sets}
\label{sec:eval_sched}

\begin{figure*}[t]
\centering
\subfloat[]{
    \includegraphics[width=0.31\linewidth]{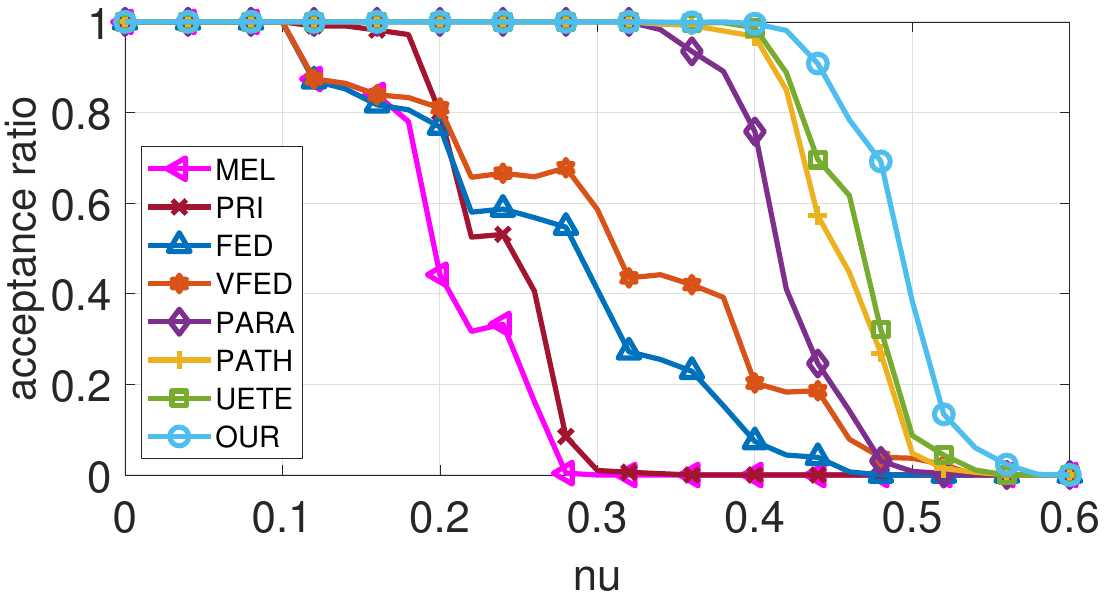}
    \label{fig:nu_16}
}
\hfil
\subfloat[]{
    \includegraphics[width=0.31\linewidth]{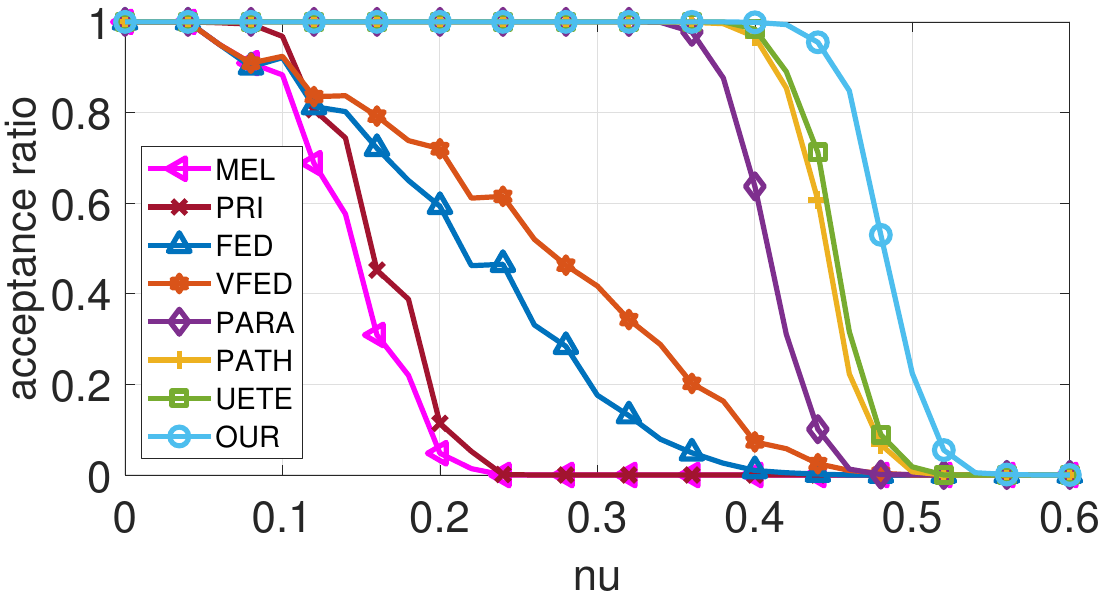}
    \label{fig:nu_32}
}
\hfil
\subfloat[]{
    \includegraphics[width=0.31\linewidth]{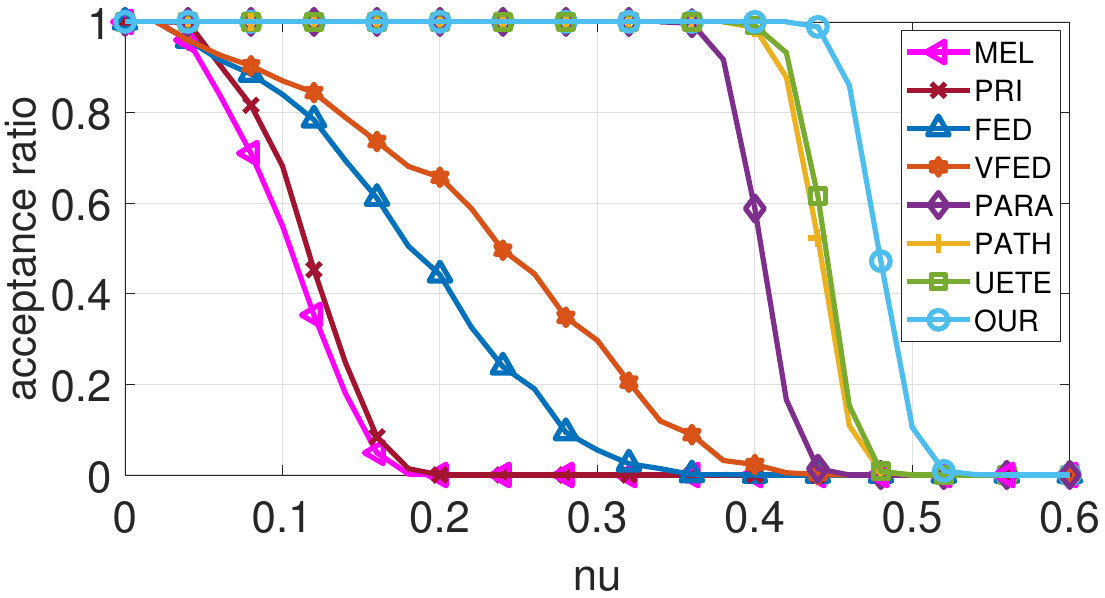}
    \label{fig:nu_64}
}
\caption{{\diff Acceptance ratio with changing the normalized utilization. (a) $m=16$, $\mathit{pf} \in [0.1, 0.6]$. (b) $m=32$, $\mathit{pf} \in [0.1, 0.6]$. (c) $m=64$, $\mathit{pf} \in [0.1, 0.6]$.}}
\label{fig:eval_sched}
\end{figure*}

This subsection evaluates the performance of scheduling DAG task sets. The following scheduling approaches are compared.
\begin{itemize}
    {\diff \item \textsf{MEL}. The global scheduling in \cite{melani2016schedulability}.
    \item \textsf{PRI}. The global scheduling in \cite{zhao2020dag} that makes use of vertex-level priorities.}
    \item \textsf{FED}. The federated scheduling in \cite{li2014analysis} based on Graham's bound in \cite{graham1969bounds}.
    \item \textsf{VFED}. The virtually federated scheduling in \cite{jiang2021virtually} based on Graham's bound in \cite{graham1969bounds}.
    \item \textsf{PARA}. The federated scheduling in \cite{he2023degree} based on the multi-path bound in \cite{he2023degree}.
    \item \textsf{PATH}. The federated scheduling in \cite{he2022bounding} based on the multi-path bound in \cite{he2022bounding}.
    \item \textsf{UETE}. The reservation-based scheduling in \cite{ueter2023parallel} based on the multi-path bound in \cite{ueter2023parallel}.
    \item \textsf{OUR}. The federated scheduling based on the proposed multi-path bound.
\end{itemize}

\noindent
\textbf{Explanation for} \textsf{OUR}\textbf{.}
Federated scheduling\footnote{Federated scheduling in \cite{li2014analysis} distinguishes heavy tasks and light tasks. This paper focuses on heavy tasks; the scheduling of light tasks is the same as \cite{li2014analysis}.} is a scheduling paradigm where each DAG task is scheduled independently on a set of dedicated cores.
When applying our bound into federated scheduling, the only extra effort is to decide the number of cores allocated to each task. This can be simply done by starting from $m=1$ and iteratively increasing $m$ until the resulting bound is no larger than the deadline of the task.
The application of our bound into federated scheduling is essentially the same as the application of the bound in \cite{he2022bounding} into federated scheduling. See \cite{he2022bounding} for more details.

\noindent
\textbf{Explanation for} \textsf{UETE}\textbf{.}
In \cite{ueter2023parallel}, two methods of computing reservations (i.e. the gang reservation and the ordinary reservation) are proposed and no schedulability tests for reservations are discussed. First, according to \cite{ueter2023parallel}, the performances of the two reservation-computing methods are almost the same. The gang reservation is used in our experiments as it is simpler than the ordinary reservation. Second, the schedulability test in \cite{dong2017analysis} is used for the gang reservations.

\noindent
\textbf{Remarks on Scheduling Algorithms.}
All the above-evaluated approaches are hierarchical scheduling and involve two levels: task-level and vertex-level. Task-level is about scheduling DAG tasks and vertex-level is about scheduling vertices in the DAG task. For vertex-level scheduling, as stated in Section \ref{sec:eval_bound}, \textsf{UETE} is more complex than other approaches. For task-level scheduling, federated scheduling in \textsf{FED}, \textsf{PARA}, \textsf{PATH} and \textsf{OUR} is very simple, as it simply assigns a dedicated set of cores to each task; the scheduling in \textsf{VFED} and \textsf{UETE} is more complex, as it requires delicate techniques to handle the sharing of cores among different reservations or servers.
Overall, the complexities of the scheduling in \textsf{FED}, \textsf{PARA}, \textsf{PATH} and \textsf{OUR} are the same and are much simpler than those of \textsf{VFED} and \textsf{UETE}.
The simplicity of scheduling algorithms is critical for the robustness of embedded real-time systems.
The only difference among \textsf{FED}, \textsf{PARA}, \textsf{PATH} and \textsf{OUR} is that the number of cores assigned to each DAG task is computed according to different response time bounds and is thus different.

\textsf{PARA} and \textsf{PATH} are the state-of-the-art approaches for scheduling DAG tasks. \textsf{UETE} is recently proposed, but it does not present a complete scheduling approach for DAG tasks as mentioned before, nor compare itself to other scheduling approaches. This paper compares them all. In the experiments, we use a standard metric called \emph{acceptance ratio} to evaluate the performances of different approaches. Acceptance ratio is the ratio between the number of schedulable task sets and the number of all evaluated task sets.

\noindent
\textbf{Task Set Generation.}
DAG tasks are generated by the same method as Section \ref{sec:eval_bound} with $c(v)$, $|V|$ and $\mathit{pf}$ randomly chosen in [5, 100], [150, 250], [0.1, 0.6], respectively.
The period $T$ (which equals the deadline $D$ in the experiment) is computed by $len(G)+\mathit{df}(vol(G)-len(G))$, where $\mathit{df}$ is a parameter. Same as the setting of \cite{he2022bounding, he2023degree}, we consider $\mathit{df}$ in [0, 0.5] to let each DAG task require at least two cores.
The number of cores $m$ is set to be 16, 32, 64.
The \emph{utilization} of a DAG task $G$ is defined to be $vol(G)/T$, and the utilization of a task set is the sum of all utilizations of tasks in this task set.
The \emph{normalized utilization} $\mathit{nu}$ of a task set is the utilization of this task set divided by the number of cores $m$.
To generate a task set with a specific utilization, we randomly generate DAG tasks and add them to the task set until the total utilization reaches the required value.
For each data point in Fig. \ref{fig:eval_sched}, we generate 1000 task sets to compute the acceptance ratio.

Fig. \ref{fig:eval_sched} reports the acceptance ratio of different approaches with changing the number of cores $m$ and the normalized utilization $\mathit{nu}$.
From the experiment results in Fig. \ref{fig:eval_sched}, we can observe that
{\diff first, federated scheduling outperforms global scheduling, which is widely observed in literature \cite{jiang2017semi,ueter2018reservation}.}
Second, the scheduling approaches based on multi-path bounds (i.e. \textsf{PARA}, \textsf{PATH}, \textsf{UETE} and \textsf{OUR}) perform significantly better than approaches based on Graham's bound (i.e. \textsf{FED} and \textsf{VFED}).
This is because multi-path bounds, which utilize the information of multiple paths to analyze the execution behavior of parallel tasks, can inherently leverage the power of multi-cores.
Third, the three existing scheduling approaches based on multi-path bounds (i.e. \textsf{PARA}, \textsf{PATH}, \textsf{UETE}) exhibit similar performances, especially for \textsf{PATH} and \textsf{UETE}, the two of which have almost the same performance.
This is consistent with the results reported in Fig. \ref{fig:eval_bound}.
Fourth, our approach, by lifting the constraint of the longest path and optimally computing the multi-path bound, advances the state-of-the-art regarding scheduling DAG tasks one step further.
The performance improvement is up to 53.2\% compared with \textsf{UETE} for $m=32$ in Fig. \ref{fig:nu_32}.

\section{Conclusion}
\label{sec:conclusion}
This paper investigates the multi-path bounds of DAG tasks.
We derive a new response time bound for a DAG task and propose an optimal algorithm to compute this bound.
We further present a systematic analysis on the dominance and the sustainability of three existing multi-path bounds, the proposed multi-path bound and Graham's bound.
Our bound theoretically dominates and empirically outperforms all existing multi-path bounds and Graham's bound.
Besides, the proposed bound is the only multi-path bound that is proved to be self-sustainable.

\begin{comment}
Requiring the longest path of the DAG task in the computation of the response time bound is a serious constraint when extending the idea of multi-path bounds to other more realistic task models, such as the conditional DAG task model.
This is because it can be very difficult to ensure the presence of the longest path in the generalized path list while making abstractions regarding the various conditional branches in a task.
This paper lifts the constraint of the longest path in the proposed multi-path bound, thus possibly bringing new opportunities for extending multi-path bounds to other more realistic task models.
In the future, we will explore this direction.
\end{comment}

\section*{Acknowledgments}
This work is supported by the Research Grants Council of Hong Kong under Grant GRF 15206221 and GRF 11208522.

\bibliographystyle{IEEEtran}
\bibliography{reference}

\begin{IEEEbiography}[{\includegraphics[width=1in,height=1.25in,clip,keepaspectratio]{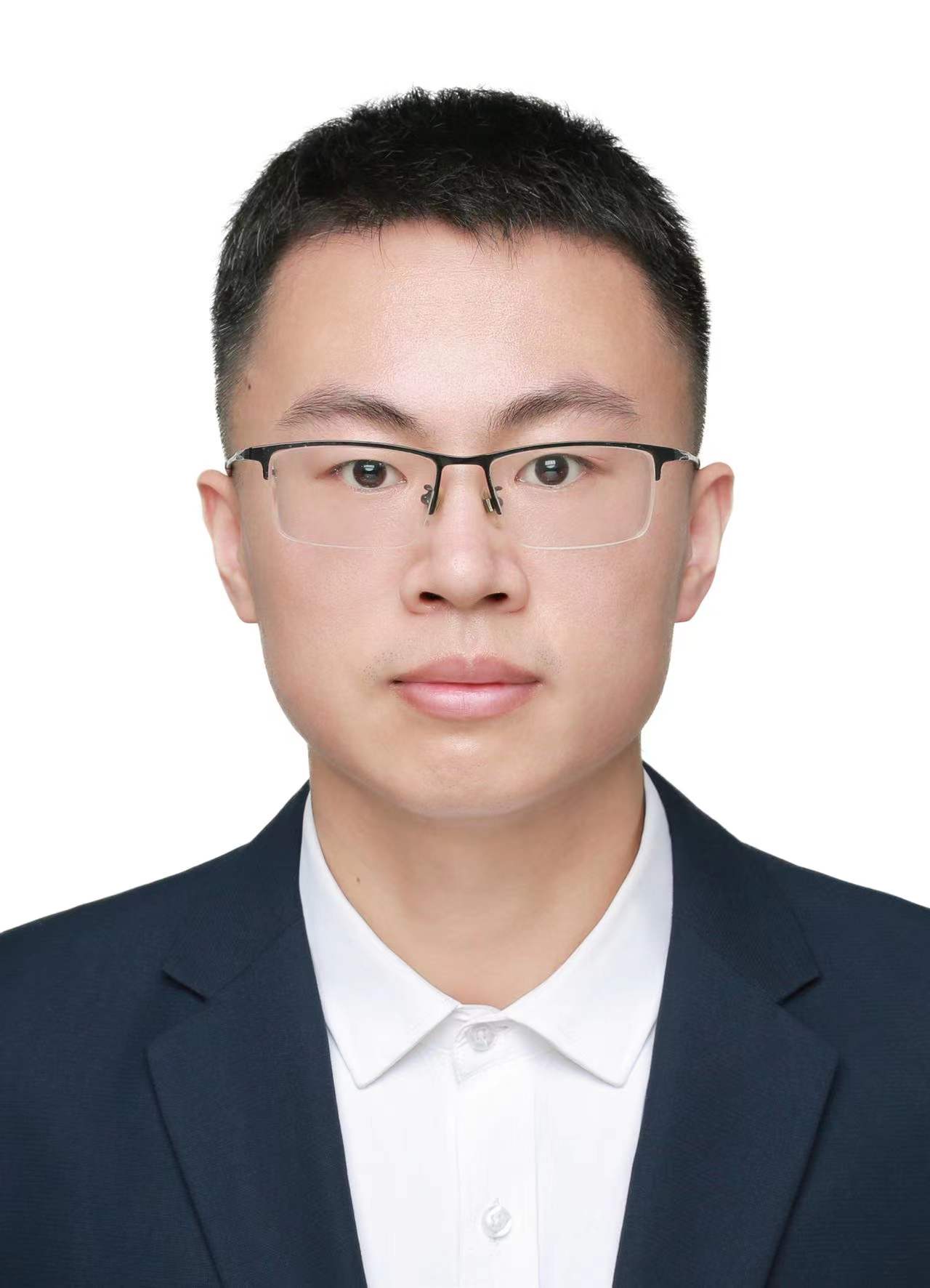}}]{Qingqiang He} is an assistant professor at Great Bay University, China. He received the Ph.D. degree in computer science from The Hong Kong Polytechnic University in 2023. His research interests include real-time scheduling theory and embedded real-time systems. He is a recipient of the Outstanding Paper Award of IEEE Real-Time Systems Symposium (RTSS) in 2022.
\end{IEEEbiography}

\begin{IEEEbiography}[{\includegraphics[width=1in,height=1.25in,clip,keepaspectratio]{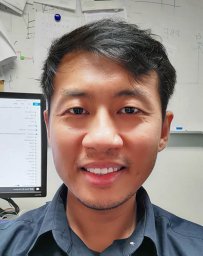}}]{Nan Guan}
is currently an associate professor at the Department of Computer Science, City University of Hong Kong. Dr. Guan received his BE and MS from Northeastern University, China in 2003 and 2006, respectively, and a Ph.D. from Uppsala University, Sweden in 2013. Before joining CityU, he worked in The Hong Kong Polytechnic University and Northeastern University, China. His research interests include real-time embedded systems and cyber-physical systems. %He received the EDAA Outstanding Dissertation Award in 2014, the Best Paper Award of IEEE Real-time Systems Symposium (RTSS) in 2009, the Best Paper Award of Conference on Design Automation and Test in Europe (DATE) in 2013.
\end{IEEEbiography}

\begin{IEEEbiography}[{\includegraphics[width=1in,height=1.25in,clip,keepaspectratio]{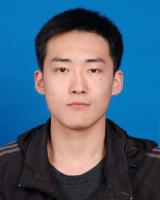}}]{Shuai Zhao}
is an associate professor at the Sun Yat-Sen University, China. He received a Ph.D. degree in computer science from the University of York in 2018. His research interests include scheduling algorithms, multiprocessor resource sharing, schedulability analysis, and safety-critical programming languages.
\end{IEEEbiography}

\begin{IEEEbiography}[{\includegraphics[width=1in,height=1.25in,clip,keepaspectratio]{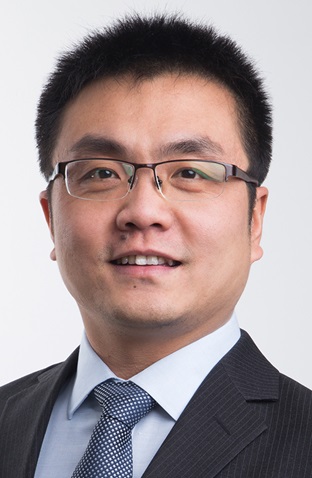}}]{Mingsong Lv}
received his Ph.D. degree in computer science from Northeastern University, China, in 2010. He is currently with the Hong Kong Polytechnic University. His research interests include timing analysis of real-time systems and intermittent computing.
\end{IEEEbiography}

\includepdfmerge{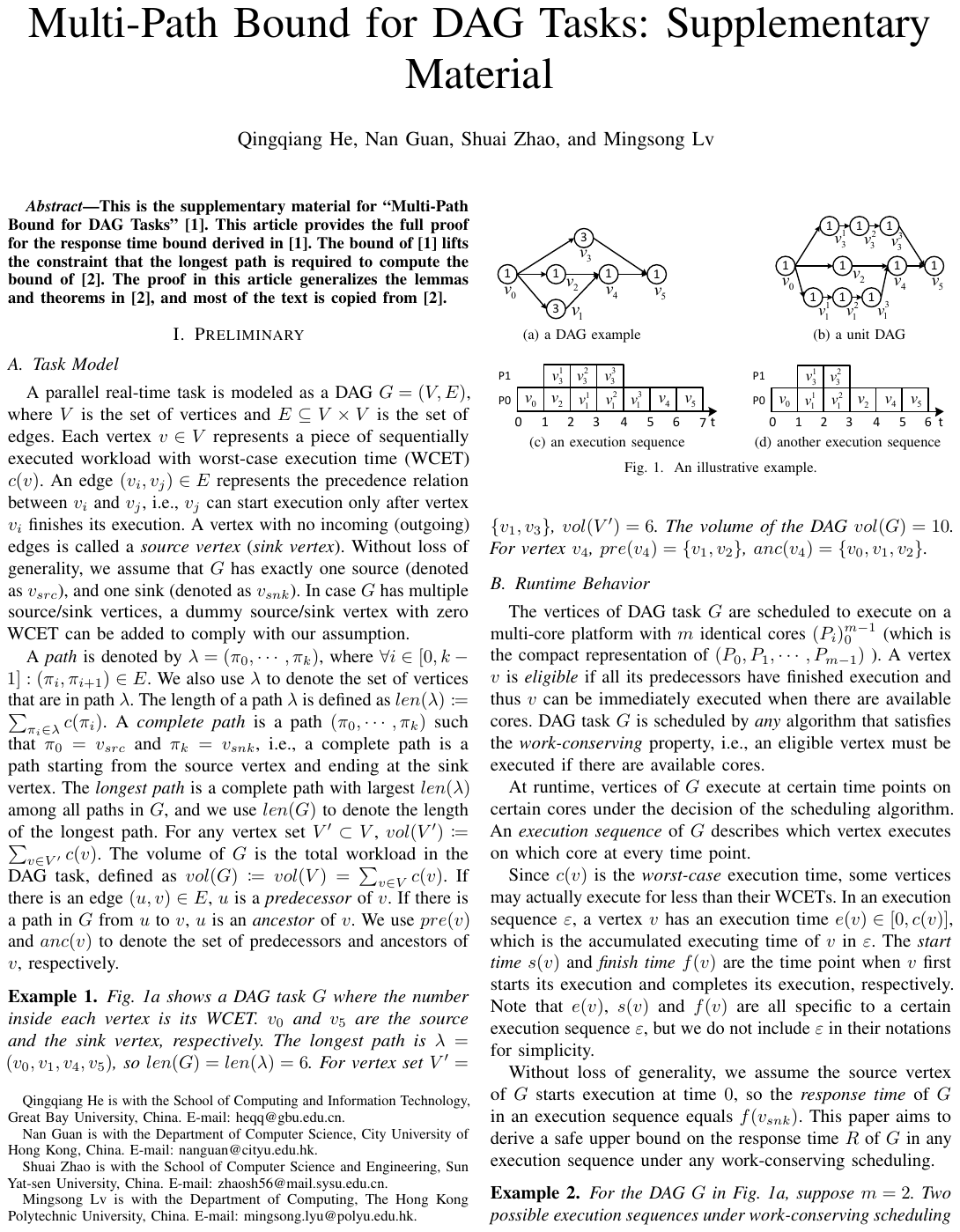,1-8}

\end{document}